\documentclass[a4paper, floatfix, 11pt, onecolumn, , notitlepage, tightenlines, superscriptaddress, accepted=2026-03-13]{quantumarticle}
\pdfoutput=1

\usepackage{algorithm}
\usepackage{algpseudocode}
\usepackage[T1]{fontenc}
\usepackage[utf8]{inputenc}
\usepackage{microtype}
\usepackage{array}
\usepackage{graphicx}
\usepackage{dcolumn}
\usepackage{bm}
\usepackage{amsmath}
\usepackage{amsfonts}
\usepackage{amssymb}
\usepackage{amstext}   
\usepackage{amsthm}
\usepackage{array}
\usepackage[english]{babel}
\usepackage{makecell}
\usepackage{dsfont}
\usepackage{nameref}
\usepackage{color}
\usepackage{float}
\usepackage[outdir=./]{epstopdf}
\usepackage{nameref}
\usepackage{subfloat}
\usepackage{hyperref}
\usepackage{cleveref}
\usepackage[figure]{hypcap}
\usepackage{multirow}
\usepackage{makecell}
\usepackage{xcolor}
\usepackage{mathtools}
\usepackage{mathrsfs}
\usepackage{wasysym}
\usepackage{mathdots}
\usepackage[notransparent]{svg}

\hypersetup{colorlinks=true, allcolors=blue}

\AtBeginDocument{
\heavyrulewidth=.08em
\lightrulewidth=.05em
\cmidrulewidth=.03em
\belowrulesep=.65ex
\belowbottomsep=0pt
\aboverulesep=.4ex
\abovetopsep=0pt
\cmidrulesep=\doublerulesep
\cmidrulekern=.5em
\defaultaddspace=.5em
\tabcolsep=7pt
}
\usepackage{booktabs}

\makeatletter
\newcommand\fs@booktabsruled{%
  \def\@fs@cfont{\bfseries\strut}\let\@fs@capt\floatc@ruled
  \def\@fs@pre{\hrule height\heavyrulewidth depth0pt \kern\belowrulesep}%
  \def\@fs@mid{\kern\aboverulesep\hrule height\lightrulewidth\kern\belowrulesep}%
  \def\@fs@post{\kern\aboverulesep\hrule height\heavyrulewidth\relax}%
  \let\@fs@iftopcapt\iftrue
}
\makeatother
\floatstyle{booktabsruled}
\restylefloat{algorithm}

\usepackage{etoolbox}
\AtBeginEnvironment{algorithm}{\linespread{1.25}\selectfont}

\definecolor{emerald}{rgb}{0.07, 0.53, 0.03}

\newcommand{\avg}[1]{\ensuremath{\left\langle#1\right\rangle}}
\newcommand{\bra}[1]{\ensuremath{\left\langle#1\right|}}
\newcommand{\ket}[1]{\ensuremath{\left|#1\right\rangle}}
\newcommand{\ketbra}[2]{\ket{#1}\!\!\bra{#2}}

\newcommand{\mbf}[1]{\mathbf{#1}}
\newcommand{\mbx}{\mathbf{x}}
\newcommand{\mby}{\mathbf{y}}
\newcommand{\mge}{\xi}

\usepackage{amsbsy}
\newcommand{\bs}[1]{\boldsymbol{#1}}
\newcommand{\bp}{\mathrm{P}}
\newcommand{\bfp}{\mathbf{P}}

\newcommand{\mm}[1]{\gamma_{#1}}
\newcommand{\mc}[1]{\mathcal{#1}}
\newcommand{\mcu}{\mathcal{U}}
\newcommand{\mcv}{\mathcal{V}}

\newcommand{\mce}{\mathcal{E}}
\newcommand{\mcs}{\mathcal{S}}

\renewcommand{\ss}[1]{{\sf{#1}}}
\newcommand{\ds}[1]{\mathds{#1}}
\newcommand{\bsalph}{\boldsymbol{\alpha}}
\newcommand{\bsbet}{\boldsymbol{\beta}}
\newcommand{\bset}{\boldsymbol{\eta}}

\makeatletter
\newtheorem*{rep@theorem}{\rep@title}
\newcommand{\newreptheorem}[2]{%
\newenvironment{rep#1}[1]{%
 \def\rep@title{#2 \ref{##1}}%
 \begin{rep@theorem}}
 {\end{rep@theorem}}}
\makeatother

\newtheorem{theorem}{Theorem}
\newtheorem{lemma}[theorem]{Lemma}
\newreptheorem{lemma}{Lemma}

\newtheorem{definition}[theorem]{Definition}

\newtheorem{innercustomresult}[theorem]{Result}

\usepackage{subfiles}

\begin{document}

\title{Fermionic Averaged Circuit Eigenvalue Sampling}
\author{Adrian Chapman}
\affiliation{Phasecraft Inc., Washington DC, USA}
\author{Steven T. Flammia}
\affiliation{Phasecraft Inc., Washington DC, USA}
\affiliation{Department of Computer Science, Virginia Tech, Alexandria, VA, USA}

\begin{abstract}
    Fermionic averaged circuit eigenvalue sampling (FACES) is a protocol to simultaneously learn the averaged error rates of many fermionic linear optical (FLO) gates simultaneously and self-consistently from a suitable collection of FLO circuits. 
    It is highly flexible, allowing for the in situ characterization of FLO-averaged gate-dependent noise under natural assumptions on a family of continuously parameterized one- and two-qubit gates. 
    We rigorously show that our protocol has an efficient sampling complexity, owing in-part to useful properties of the Kravchuk transformations that feature in our analysis. 
    We support our conclusions with numerical results. 
    As FLO circuits become universal with access to certain resource states, we expect our results to inform noise characterization and error mitigation techniques on universal quantum computing architectures which naturally admit a fermionic description.
\end{abstract}

\maketitle

\section{Introduction}

As fundamental building blocks of our universe, fermions are naturally the mediators of quantum information in many quantum computing architectures, such as superconducting \cite{arute2019quantum} and semiconductor devices \cite{burkard2023semiconductor}.
Many of the systems we hope to simulate with a quantum computer are systems of interacting fermions.
These include quantum materials and quantum chemical systems, where the relevant degrees of freedom are interacting electrons \cite{mcardle2020quantum}. 
It is thus natural to consider models of quantum computation that regard fermions as central objects.
In principle, this removes any overhead from switching between fermion- and spin-based encodings, but it also presents unique challenges for other quantum information processing tasks due to the nonlocal nature of fermionic antisymmetric exchange statistics.

Predominant among these tasks is the characterization of quantum noise. 
Whereas noise is usually assumed to act independently on local degrees of freedom, subsequent gates will cause the effect of any such noise process to spread.
In particular, fermionic models of quantum computing generally involve families of continuously parameterized gates. 
In this setting, even local noise will become complex---and thus difficult to learn---on the circuit timescales necessarily to realize quantum computational advantage.
To address this issue, we can take inspiration from solutions in the qubit setting. 
In particular, we consider the strategy of \emph{noise tailoring}, whereby noise is first physically driven to a desired form via averaging, or \emph{twirling}, and the averaged noise parameters can be more easily learned ~\cite{knill2005quantum, wallman2016noise}.

Noise tailoring is well-suited to the setting of gate-dependent noise, where we assume each distinct physical operation on the experiment introduces a unique, reproducible noise channel. 
The general problem of estimating all noise in a gate set having gate-dependent noise \textit{without} any averaging is quite challenging~\cite{blumekohout2013robust, kim2015microwave, nielsen2021gate}, so it is natural to ask what role noise tailoring can play to simplify the problem. 
Ref.~\cite{flammia2022averaged} introduced averaged circuit eigenvalue sampling (ACES) for simultaneously characterizing the gate-dependent noise of many noisy Clifford gates with Pauli twirling. 
Clifford circuits with Pauli noise preserve the set of Pauli observables up to an overall factor---called the circuit eigenvalue---which quantifies the noise in the circuit.
The main idea behind ACES is that the eigenvalues of many different noisy circuits are all dependent on a common set of gate-level noise eigenvalues, which can be extracted with access to enough experimental data.

It is well-understood that Clifford circuits with stabilizer measurements are not universal for quantum computation.
Any attempt to extend the ACES protocol to a universal circuit family will thus require a generalization to circuits and noise models that do not preserve unnormalized Pauli operators.
In this work, we provide a protocol for the characterization of noisy \textit{fermionic linear optical} (FLO) circuits, or equivalently, their qubit-embedded versions, noisy nearest-neighbor \textit{matchgate} circuits.
We call this FACES, for \textit{fermionic averaged circuit eigenvalue sampling}. 
FLO circuits describe general time-dependent free-fermion dynamics, and matchgate circuits describe these same dynamics on qubits under a Jordan-Wigner transformation~\cite{jozsa2008matchgates}. 
We work almost exclusively in the fermion picture here, but our results are interchangable across the two paradigms. 

Though FLO and matchgate circuits are comprised of non-Clifford gates, they are nonetheless efficiently simulable classically. 
However, they become universal for quantum computation when we additionally include the application of arbitrary Clifford gates; for example, even adding nearest-neighbor SWAP operations between qubits is sufficient.
Similarly, Clifford circuits become universal for quantum computation with the inclusion of arbitrary matchgates.
Since characterizing gate-dependent noise in Clifford circuits is a well-studied problem, we expect that FACES can be used in tandem with other methods to estimate averaged noise in universal gate sets. 
Furthermore, we expect that FACES will find natural applications in the setting of quantum materials and chemistry simulations, for which fermionic degrees of freedom are fundamental.

Our main results include an algorithm for extracting the FLO-averaged eigenvalues of the elements of an elementary set of FLO unitary gates.
Additionally, we rigorously prove that our protocol has an efficient sample complexity when reasonable assumptions about our setting are met.
Finally, we demonstrate the efficacy of our protocol by numerically showing that it is feasible for realistic error rates.

The input to our protocol is a model, which we define as follows
\begin{definition}[Model]
    A FACES model on $2n$ fermionic modes is specified by a 3-tuple
    \begin{align}
        \mc{M} = (\{G_1,\ldots,G_K\}, \{\boldsymbol{\xi}_{1}, \dots, \boldsymbol{\xi}_{K}\}, \{C_1,\ldots,C_J\})
    \end{align}
    for $K$ gates $G_k$, each of which is associated to a vector of $2n + 1$ channel eigenvalues $\bs{\xi}_k$, and $J \geq K$ FLO circuits $C_j$. 
    From these, we compute the design matrix $\mbf{A}$, whose integer-valued element $A_{jk}$ specifies the number of occurrences of gate $G_k$ in circuit $C_j$.
    A FACES experiment constructs this matrix and samples each circuit $S$ times each.
    We say that a model is realizable by our experiment if the noisy gates are described by a set of true values for $\bs{\xi} \coloneqq \cup_k \ \bs{\xi}_k$.
    \label{def:model}
\end{definition}

Our main result applies to the realizable setting (i.e.\ we don't cover the \emph{agnostic setting}, where we allow for the possibility that the true gate noise is not represented by the model).
It is stated informally as follows.

\begin{innercustomresult}[Informal Statement of \cref{thm:samplecomplexity}]
    Consider the FACES model $\mc{M} = (\mc{G}, \bs{\xi}, \mc{C})$.
    A FACES experiment on $\mc{M}$ returns an estimate $\hat{\boldsymbol{\xi}}$ for the true gate eigenvalues $\boldsymbol{\xi}$ satisfying
    \begin{align}
    \lVert \hat{\bs{\xi}} - \bs{\xi} \rVert_{\infty} \leq \lVert \mbf{A}^+ \rVert_{\infty}\,\varepsilon,
    \label{}
    \end{align}
    where $\mbf{A}^+$ is the pseudoinverse of the design matrix, with $S = O\bigl(\frac{n \log(J)}{\varepsilon^2}\bigr)$ samples per circuit.
\end{innercustomresult}

The norm on $\mbf{A}^+$ is the induced $\infty\to\infty$ norm, defined as $\lVert X \rVert_{\infty} = \mathrm{max}_i \sum_{j} |X_{ij}|$.
In the setting where $\mbf{A}$ does not have full column-rank, its psuedoinverse is still well-defined, but the model includes at least one unlearnable, or ``gauge'' parameter that is inaccessible to our framework (for a detailed investigation of this, see \cite{chen2023learnability,chen2024efficient}).

Considering eq. (2), one imagines minimizing $\rVert \mbf{A}^+ \lVert_{\infty}$ as a promising method to likely improve the precision of the estimate.
Previous work has considered optimizing gate-dependent noise learning in a restricted setting (this is done in gate-set tomography by amplifying certain errors using long sequences \cite{nielsen2021gate}, for example).
Other methods have also been explored in an ACES-type context.
For example, the authors of ref.~\cite{hockings2025scalable} consider optimizing a figure of merit related to the gate-eigenvalue estimator covariance matrix.
A comprehensive theory is lacking and currently constitutes an open research direction.

The characterization of noisy fermionic linear optical gates is well-studied in the literature \cite{claes2021character, liu2023group, wilkens2024benchmarking, cudby2024learning}.
Of particular relevance to our works are the results of refs.~\cite{helsen2022matchgate} and \cite{burkat2024lightweight}. 
In ref.~\cite{helsen2022matchgate}, the authors propose a method for estimating the average fidelity of noisy matchgate circuits.
In \cite{burkat2024lightweight}, the authors introduce a corresponding protocol that does not require an average over the matchgate group.
Our result extends the former result to the in situ case under similar physical requirements, though it still requires a group average.
Specifically, we require the ability to randomly sample unitaries to perform the FLO twirling and that the effect of the noise on the twirling gates can itself be incorporated into our framework.
We return to this point throughout the paper.

We also require that we are able to reliably prepare the all-zeros state $\ket{\mbf{0}}$ and all-plus state $\ket{+}^{\otimes n}$, and that we generate a set of noisy FLO circuits that act as a particular target unitary on-net (either the identity, or a particular unitary $U_+$, to be defined). 
We furthermore need to be able to measure in the computational basis and the Pauli $Y$-basis on at least one qubit. 

\begin{figure}
    \centering
    \includegraphics[width=\textwidth]{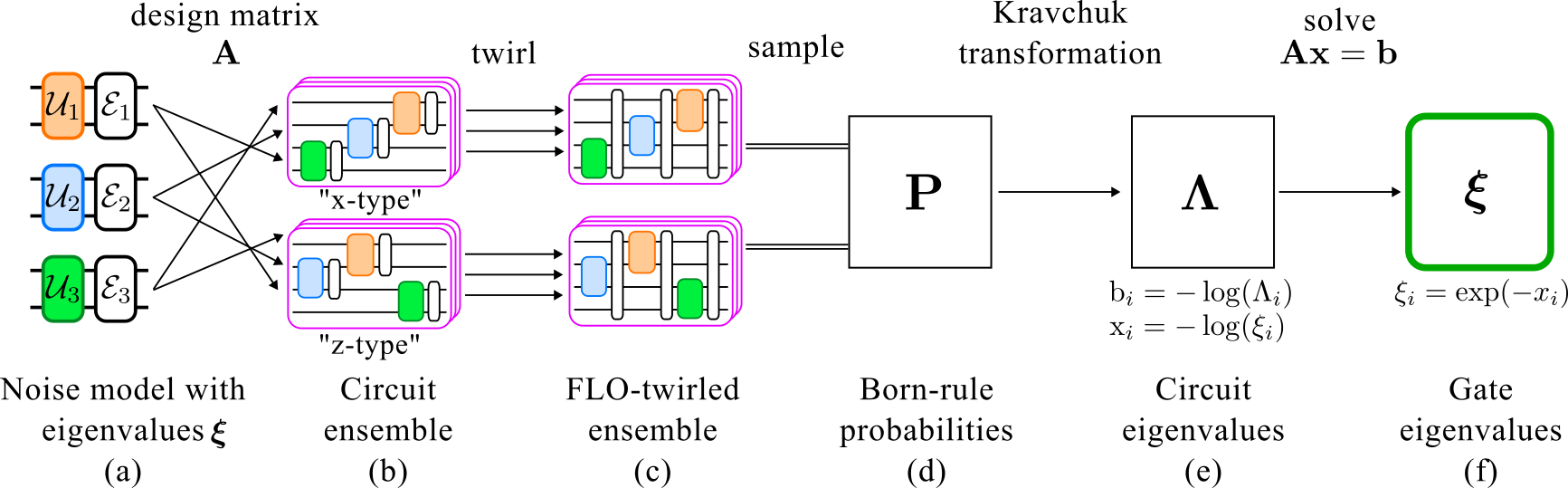}
    \caption{A graphical outline of the FACES protocol. (a) Given a set of noisy gates $\{\mcu_i \mce_i\}_i$ we wish to characterize, we model each such gate as the associated noise channel $\mce_i$, followed by an application of the ideal gate $\mcu_i$. (b) We compose the gates into a collection of circuits according to the design matrix $\mbf{A}$. (c) For each circuit $\mathcal{C}$ in the collection, we randomly sample from an ensemble of circuits whose average is $\mathcal{C}$ subject to matchgate-twirled noise. (d) We sample from the circuits in our collection to obtain an estimate for the Born-rule probabilities of the associated outcomes. (e) We apply a Kravchuk transformation to obtain the eigenvalues of each circuit. (f) We fit the gate eigenvalues to the data using a log-linear model.}
    \label{fig:controlflow}
\end{figure}

Our paper is organized as follows: in \cref{sec:background}, we detail the relevant background on fermionic linear optical circuits (\cref{sec:matchgates}), Kravchuk matrices (\cref{sec:kravchuk}), twirling (\cref{sec:mgn}), and Pauli channel eigenvalues (\cref{sec:channeleig}).
In \cref{sec:protocol}, we explain our measurement protocol, including our assumptions (\cref{sec:modelassumptions} and \cref{sec:floassumptions}), stages of the protocol (\cref{sec:protocoldetails}) and sample complexity (\cref{sec:complexity}).
In particular, we include a proof of our sample complexity guarantee in the main text.
Finally, we provide a numerical demonstration of our protocol in \cref{sec:numerics} and conclude in \cref{sec:conclusion}.
Our protocol in pseudocode format is detailed in \cref{sec:pseudocode}, and remaining proofs are given  in \cref{sec:proofs}.

\section{Background and Notation}
\label{sec:background}

\subsection{Fermionic Linear Optics} 
\label{sec:matchgates}
Fermionic linear optical (FLO) circuits constitute the most general group of unitary circuits that describe free-fermion dynamics.
This group is isomorphic to the group of matchgate circuits on many-body qubit systems by the Jordan-Wigner transformation.
For our purposes, we mainly keep the discussion phrased in terms of fermionic modes, with the understanding that all of our results admit corresponding descriptions in terms of qubits.

We consider FLO circuits acting on even-many fermion modes.
For such a system of $2n$  modes, define the set of Majorana operators satisfying the canonical anticommutation relations
\begin{align}
    \{\mm{\mu}, \mm{\nu}\} = 2\delta_{\mu \nu} I \mathrm{,} \mbox{\hspace{5mm}} \mathrm{and} \mbox{\hspace{5mm}} \mm{\mu}^{\dagger} = \mm{\mu}^{\vphantom{\dagger}}
    \label{eq:car}
\end{align}
for all $\mu$, $\nu \in [2n]$, where the notation $[n]$ means the set of integers $\{1,\ldots,n\}$.
Denote the associated group of FLO circuits as $\ds{F}(2n)$.
The elements of this group are Gaussian unitaries up to reflection $U \coloneqq \gamma_{1}^{x} \exp(\sum_{\mu, \nu} h_{\mu \nu} \gamma_{\mu} \gamma_{\nu})$ for $x \in \{0, 1\}$ and $\mbf{h}$ a real antisymmetric matrix. 
These preserve linear combinations of Majorana modes under conjugation.
Writing $\mc{U}(A) \equiv U A U^{\dagger}$ as the unitary conjugation channel by an element $U \in \ds{F}(2n)$, we have 
\begin{align}
    \mc{U}(\mm{\mu}) = \sum_{\nu = 1}^{2n} R_{\mu \nu}\mm{\nu}
    \label{eq:singleparticle}
\end{align}
for $\mbf{R} = \mbf{R}^x_1 e^{-4 \mbf{h}} \in O(2n)$ the single-particle transition matrix and $\mbf{R}_1$ the canonical reflection with elements $(\mbf{R}_1)_{\mu \nu} = -(-1)^{\delta_{\mu, 1}} \delta_{\mu \nu}$. 
By the canonical anticommutation relations, \cref{eq:car}, FLO circuits also preserve fixed-degree polynomials of Majorana operators.
For a subset $\bsalph \subseteq [2n]$, let $\mm{\bsalph} \equiv \prod_{\mu \in \bsalph} \mm{\mu}$ denote the \emph{Majorana monomial} labeled by $\bsalph$, where the order of the product is taken such that the index is ascending from left to right.
Eq.~(\ref{eq:singleparticle}) generalizes as
\begin{align}
    \mc{U}(\gamma_{\bsalph}) = \sum_{\{\bsbet \subseteq [2n] | |\bsbet| = |\bsalph| \}} \mathrm{det}(\mbf{R}_{\bsalph \bsbet}) \gamma_{\bsbet}
    \label{eq:configtransform}
\end{align}
where $\mbf{R}_{\bsalph \bsbet}$ denotes the submatrix of $\mbf{R}$ with rows taken from $\bsalph$ and columns taken from $\bsbet$ (that is, Majorana monomials transform by a \emph{compound} of $\mbf{R}$).
We adopt the notational shorthand
\begin{align}
    \bsalph | k, \ell \rightarrow \{\bsalph \subseteq [\ell] | |\bsalph| = k \}
    \label{eq:notatshort} \,.
\end{align} 
Linear transformations of the form in \cref{eq:configtransform} constitute a representation of $O(2n)$ by the Cauchy-Binet formula, and this representation is irreducible.
As explained in ref.~\cite{helsen2022matchgate}, FLO transformations are uniquely determined up to an overall phase by their corresponding element of $O(2n)$.

It is convenient to generalize our index-subset notation in \cref{eq:notatshort} to binary strings with a fixed Hamming weight. 
We refer to a binary string by a bold alphabetic character, and, using the same notation as in \cref{eq:notatshort}, we denote $\mbf{x}|k, \ell \rightarrow \{\mbf{x} \subseteq \mathds{Z}_2^{\ell}| |\mbf{x}| = k\}$.
In both use cases, we drop explicit reference to $\ell$ where clear from context and simply write ``$\bsalph | k$'' or ``$\mbf{x}|k$''.

We also need to consider these objects in terms of Pauli operators.
Denote by $\ds{P}(n)$ the $n$-qubit Pauli group, and let $\ss{P}(n)$ denote this group modulo an overall phase. 
We can label an $n$-qubit Pauli operator (with a convenient choice of phase) by a $2n$-bit string $\mbx = (\mbf{a}, \mbf{b})$ as $\sigma_{\mbx} = i^{\mbf{a}\cdot \mbf{b}} \otimes_{j = 1}^{n} X_j^{a_j} Z_j^{b_j}$, denoting single-qubit Pauli operators on a fixed site by $\{X_j, Y_j, Z_j\}$. 
Since we are exclusively concerned with Paulis modulo phase, multiplication in $\ds{P}(n)$ maps conveniently to mod-2 addition of the associated binary strings in $\ss{P}(n)$. 

As stated earlier, Majorana monomials and Pauli operators are in one-to-one correspondence (up to a phase) by the Jordan-Wigner transformation
\begin{align}
\begin{cases}
    \mm{2j - 1} = Z^{\otimes (j - 1)} \otimes X_j & \\
    \mm{2j} = Z^{\otimes (j - 1)} \otimes Y_j &
    \end{cases} \mathrm{.}
    \label{eq:jwtransformation}
\end{align}
These operators satisfy the canonical anticommutation relations \cref{eq:car}.
In this context, the group of FLO circuits is isomorphic to the group of matchgate circuits, which we denote as $\ds{M}(n)$.
Matchgate circuits can be decomposed into 2-qubit matchgates of the form
\begin{align}
    G_j(A, B) \coloneqq \begin{pmatrix}
    a_{00} & 0 & 0 & a_{01} \\
    0 & b_{00} & b_{01} & 0 \\
    0 & b_{10} & b_{11} & 0 \\
    a_{10} & 0 & 0 & a_{11}
    \end{pmatrix}
\end{align}
on nearest-neighbor qubits $\{j, j + 1\}$ on a 1-d line.
Namely, $A$ acts on the even-parity 2-qubit subspace of these qubits, and $B$ on the odd-parity subspace.
We require $\mathrm{det}(A) = \mathrm{det}(B)$ for $G_j(A, B)$ to be a matchgate.
Note that, in particular, single-qubit rotations of the form $e^{i \theta Z_j}$ are matchgates for all $\theta \in [0, 2\pi)$ and all $j \in [n]$. 
Where clear from context, we refer to the groups $\ds{F}(2n)$ and $\ds{M}(n)$ interchangeably.
For example, by $\mc{U}(\sigma_{\mbf{x}})$ for $\mc{U} \in \ds{F}(2n)$, we mean conjugation by the matchgate circuit associated to $\mc{U}$ under the Jordan-Wigner transformation.

\subsection{Kravchuk Matrices}
\label{sec:kravchuk}
A central object for our analysis is the following matrix
\begin{definition}{(Kravchuk matrix)}
    The \emph{Kravchuk matrix} of order $\ell$, denoted by $\mbf{M}^{(\ell)} \in \mathds{Z}^{(\ell + 1) \times (\ell + 1)}$, is an integer-valued matrix whose elements are symmetric polynomials evaluated on $\pm 1$.
\begin{align}
    M^{(\ell)}_{jk} \coloneqq e_k( -\mbf{1}_j \oplus \mbf{1}_{\ell - j}) \mathrm{,}
    \label{eq:mdef}
\end{align}
where $j$, $k \in [\ell] \cup \{0\},$ $e_k(\boldsymbol{v})$ is the degree-$k$ elementary symmetric polynomial in the elements of the array $\boldsymbol{v}$, and $\mbf{1}_j$ is the length-$j$ vector of all ones.
\end{definition}
The Walsh-Hadamard matrix of size $2^n$ has matrix elements $(-1)^{\mbf{x}\cdot \mbf{y}}$ in the standard basis, where $\mbf{x}$ and $\mbf{y}$ are bit-strings of length $n$ that label matrix elements. 
The elements of the Kravchuk matrix are obtained by summing $\sum_{\mbf{y} | k} (-1)^{\mbf{x}\cdot \mbf{y}}$ to get a matrix element that depends only on the Hamming weight $k$ and the bit-string $\mbf{x}$. 
In fact, this expression depends only on the Hamming weight of $\mbf{x}$. 
By keeping one copy of the resulting sum for each value that $|\mbf{x}|$ can take, we get the Kravchuk matrix elements in \cref{eq:mdef}. 

The Kravchuk matrix plays an analogous role in FACES as the Walsh-Hadamard matrix does in the original (Clifford) ACES protocol. 
Namely, it is the essential link between the eigenvalues and error probabilities of a FLO-twirled channel, as explained below. 

We now give some alternative expressions for $\mbf{M}^{(\ell)}$ which clearly show that this matrix can be computed efficiently. 
Let us introduce the notation $\mathrm{coeff}_k\bigl(p(x)\bigr)$ for the coefficient of the degree-$k$ monomial term of the polynomial $p(x)$, where the variable $x$ will be understood from context. 
Then we have the following lemma. 

\begin{lemma}[Kravchuk matrix relations]
We have the alternate expressions for $M^{(\ell)}_{jk}$
\begin{align}
    M_{jk}^{(\ell)} &= \mathrm{coeff}_k\bigl((1 - u)^j (1 + u)^{\ell - j} \bigr)  \label{eq:polyexpansion} \\
    &= \sum_{\bsalph|k, \ell} (-1)^{|\bsalph \cap [j]|} \label{eq:binsum} \\
    M_{jk}^{(\ell)} &= \sum_{m = \max{(0, k + j - \ell)}}^{\min{(j, k)}} \binom{\ell - j}{k - m} \binom{j}{m} (-1)^m\,.
    \label{eq:binomexpansion}
\end{align}
Additionally, $\mathbf{M}^{(\ell)}$ is proportional to an involution
\begin{align}
    (\mathbf{M}^{(\ell)})^2 = 2^{\ell} \mathbf{I} \mathrm{.}
    \label{eq:involution}
\end{align}
\label{lem:mcoeffs}
\end{lemma}
\begin{proof}
The proof is the content of \cref{sec:erelations}.
\end{proof}

While a proof of the above relations can be found in existing literature, we include it here in the interest of keeping our results self-contained.
Where clear from context, we drop the label $\ell$ from the Kravchuk matrix and any related matrices.

In some cases, we need to account for phases that accumulate from commuting Majorana operators past one another. 
To this end, it is convenient to define a relative of the Kravchuk matrix as follows. 
\begin{definition}{(Antipode involution)}
    The \emph{antipode} involution matrix, $\mbf{s}^{(\ell)}$ for $\ell$ even, is the permutation matrix with components
    \begin{align}
        s^{(\ell)}_{jk} = \delta_{jk} \delta_{j \ (\mathrm{mod} \ 2), 0}  + \delta_{
    \ell - j, k} \delta_{j \ (\mathrm{mod} \ 2), 1} \mathrm{.}
    \end{align}
\end{definition}
The antipode involution takes indices with fixed odd weight and exchanges $j \leftrightarrow \ell-j$. 
It satisfies the following relations due to the particular symmetry properties of $\mbf{M}$
\begin{lemma}{(Kravchuk-antipode relations)}
The Kravchuk matrix and antipode involution commute, and their product has components
\begin{align}
    (\mbf{s}^{(\ell)} \cdot \mbf{M}^{(\ell)})_{jk} = (\mbf{M}^{(\ell)} \cdot \mbf{s}^{(\ell)})_{jk} = (-1)^{jk} e_k(-\mbf{1}_j \oplus \mbf{1}_{\ell - j}) \mathrm{.}
\label{eq:permmdef}
\end{align}
\label{lem:karelations}
\end{lemma}
\begin{proof}
The proof is the content of \cref{sec:kaproof}.
\end{proof}

Finally, we define the diagonal matrix of binomial coefficients $\mbf{d}^{(\ell)}$ by
\begin{align}
    d^{(\ell)}_{jk} = \binom{\ell}{k} \delta_{jk}\mathrm{.}
\end{align}
It is straightforward to see that $\mbf{s}^{(\ell)}$ and $\mbf{d}^{(\ell)}$ commute as well, but $\mbf{M}^{(\ell)}$ and $\mbf{d}^{(\ell)}$ do not commute.
Each of these transformations will be important for our estimation procedure.

\subsection{Noise Channels and Twirling}
\label{sec:mgn}
A quantum noise channel is described by a completely positive trace-preserving (CPTP) map, admitting a Kraus representation
\begin{align}
    \mce(X) = \sum_j K^{\vphantom{\dagger}}_j X K_j^{\dagger},
    \label{eq:krauszdef}
\end{align}
where $\{K_j\}$ are the Kraus operators satisfying \begin{align}\sum_j K_j^{\dagger} K^{\vphantom{\dagger}}_j = I \mathrm{.}
\end{align}
For a unitary ensemble $\mcs$ (either discrete or continuous), we define the $\mcs$-twirling of $\mce$ as
\begin{align}
    \mce^{\mcs} = \mathbb{E}_{\mcu \in \mcs} \bigl[\mc{U} \circ \mce \circ \mc{U}^{\dagger}\bigr]
\end{align}
where we employ the unitary channel notation from \cref{sec:matchgates}, and $\mathbb{E}_{x}[f(x)]$ denotes the expectation value over $x$ of $f(x)$.
The main examples we consider are Pauli and FLO twirling, where the expectation is given respectively by
\begin{align}
\mce^{\ds{P}(n)}(X) &= 4^{-n} \sum_{\mbx \in \ss{P}(n)} \sum_{j} \left(\sigma_{\mbx} K_j \sigma_{\mbx} \right) X \left(\sigma_{\mbx} K_j \sigma_{\mbx} \right)^{\dagger}\,,
\label{eq:ptwirlkraus}
\end{align}
and
\begin{align}
    \mce^{\ds{F}(2n)}(X) &= \int_{U \in \ds{F}(2n)} dU \ \sum_j \mc{U}\left(K_j\right) X \ \mc{U}\left( K_j \right)^{\dagger}
    \label{eq:mgtwirldef}
\end{align}
where \cref{eq:mgtwirldef} is a Haar average over the group of FLO unitaries. 
Note also that twirling is idempotent. 
That is, twirling twice is the same as twirling once. 
Additionally, $\mce^{\mcs}$ commutes with every element of $\mcs$ when $\mcs$ is a group, by construction. 

A channel that is invariant under Pauli twirling is called a Pauli channel. 
It is easy to show that Pauli twirling a general channel results in a Pauli channel~\cite[Lemma 5.2.4]{dankert2005efficient}.
Pauli channels can be parameterized as 
\begin{align}
\label{eq:ptwirl}
    \mce^{\ds{P}(n)}(X) &= \sum_{\mbf{x} \in \ss{P}(n)} p_{\mbf{x}} \sigma_{\mbf{x}} X \sigma_{\mbf{x}} 
\end{align}
where the $\{p_{\mbf{x}}\}_{\mbf{x} \in \ss{P}(n)}$ form a probability distribution over bit-strings of length $2n$. 
These are called the \textit{Pauli error probabilities}. 
Starting from a general channel with Kraus operators $K_j \equiv \sum_{j} c_{j\mbf{x}} \sigma_{\mbf{x}}$ as in \cref{eq:krauszdef}, the relationship to the Pauli error probabilities of the twirled channel in \cref{eq:ptwirl} is~\cite[Lemma 5.2.4]{dankert2005efficient}
\begin{align}
    p_{\mbf{x}} = \sum_{j} |c_{j\mbf{x}}|^{2}\,.
    \label{eq:krausprob}
\end{align}
 
Similarly, a \emph{FLO-twirled channel} is one that is invariant under FLO twirling.
Since all Pauli operators are matchgates, any FLO-twirled channel is also invariant under Pauli twirling, and is therefore a Pauli channel.
Note the following contrast between Pauli twirling and FLO twirling: Pauli twirling a unitary Pauli (i.e. rank-1)  channel leaves that channel invariant, but it is \emph{not} the case that FLO unitary channels are invariant under FLO twirling since they generally become non-unitary after twirling. 
FLO-twirled channels on $2n$ modes can be parameterized as 
\begin{align}
    \mce^{\ds{F}(2n)}(X) &= \sum_{k = 0}^{2n} \binom{2n}{k}^{-1} q_{k} \sum_{\bsalph | k, 2n} \mm{\bsalph}^{\vphantom{\dagger}} X \mm{\bsalph}^{\dagger}\,,
    \label{eq:mgtwirl}
\end{align}
where $q_k$ form a probability distribution over $\{0,\ldots,2n\}$. 
We call the $q_k$ the \emph{fermionic error probabilities}. 

\begin{lemma}[FLO twirl of a Pauli channel]
Consider a Pauli channel $\mce^{\ds{P}(n)}$ with Pauli error probabilities $p_{\mathbf{x}}$. 
Let $\mbx(\bsalph)$ denote the Pauli string associated to the subset of Majorana modes $\bsalph$ by the Jordan-Wigner transformation. 
Then $\mce^{\ds{F}(2n)} = (\mce^{\ds{P}(n)})^{\ds{F}(2n)}$, and the fermionic error probabilities $q_k$ of $\mce^{\ds{F}(2n)}$ are given by
\begin{align}    
q_k = \sum_{\bsalph | k} p_{\mbx(\bsalph)} \,. \label{eq:ptomgrenorm}
\end{align}
\label{lem:ptwirl}
\end{lemma}
\begin{proof}
    See \cref{sec:twirlrelations}.
\end{proof}

We remark that the FLO twirl of a \emph{general} channel has fermionic error probabilities that are related to the Kraus operators by concatenating the expressions in \cref{eq:krausprob,,eq:ptomgrenorm}. 

We model a noisy implementation $\widetilde{\mcv}$ of a FLO unitary $\mcv$ as a noise map $\mce$ followed by the ideal channel $\mcv$, with the intuition that $\mce$ is close to the identity channel. 
This provides the following
\begin{definition}{(Unitary channel noise)}
For a noisy implementation $\widetilde{\mcv}$ of a unitary conjugation channel $\mcv$, we define the noise map corresponding to $\widetilde{\mcv}$ by
\begin{align}
\mce \coloneqq \mcv^{\dagger}\widetilde{\mcv},
\end{align}
and this gives $\widetilde{\mcv} = \mcv \mce$.
\label{def:channelnoise}
\end{definition}
We will eventually require that $\mce$, so defined, is not too noisy in a sense made precise in \cref{sec:modelassumptions}. 
Prepending a random FLO unitary $\mc{U}^{\dagger}$ to $\widetilde{\mcv}$ and appending the unitary $\mc{U}$ conjugated by $\mcv$ gives the twirled noise map followed by the ideal channel in expectation.
The following statement of this is a standard result. 

\begin{lemma}[Noise Twirling]
\label{lem:noisetwirl}
For $\mc{S} \in \{\ds{P}(n), \ds{F}(2n)\}$, we have
\begin{align}
\widetilde{\mcv}^{\mcs} \coloneqq \int_{\mc{S}} d\mcu \left(\mcv \mcu \mcv^{\dagger} \right) \left(\mcv \mce \right) \mcu^{\dagger} = \mcv \mce^{\mc{S}}\,.
\label{eq:lem1line1}
\end{align}
\end{lemma}
\begin{proof}
    Immediate from the definition of twirling. 
\end{proof}

In practice, twirling is done by taking a sample average over the ensemble of unitary circuits appearing in the twirl~\cite{knill2005quantum,wallman2016noise}. 
In the Pauli-twirl case, the overhead that this induces is negligible in the sense that the depth-one Pauli unitary circuits only slightly increase the depth overall, or can often be compiled into the underlying circuit~\cite{wallman2016noise}. 
However, the best-known implementations of a generic FLO circuit require $\mc{O}(n)$ depth \cite{jiang2018quantum}.
It would be interesting to investigate reducing this depth penalty by sampling generators of the FLO group and analyzing the mixing time of the resulting random walk on $O(2n)$, similar to what was done in ref.~\cite{frana2018approximate} for finite groups, but this is beyond the scope of the present work.

\subsection{Channel Eigenvalues}
\label{sec:channeleig}

If $\mce^{\ds{P}(n)}$ is a Pauli channel, then $\mce^{\ds{P}(n)}$ satisfies an eigenvector relation with respect to the Pauli operators
\begin{align}
    \mce^{\ds{P}(n)}(\sigma_{\mbx}) = \lambda_{\mbx} \sigma_{\mbx}\,.
\end{align}
The $\lambda_{\mbx}$ are called the \emph{Pauli channel eigenvalues}. 
They are related to the Pauli error probabilities by 
\begin{align}
    \lambda_{\mbx} = \sum_{\mby \in \ss{P}(n)} (-1)^{\avg{\mbx, \mby}} p_{\mby}
    \label{eq:whtransform}
\end{align}
where $\avg{\mbx, \mby}$ is the symplectic inner product between bit strings $\mbx$, $\mby \in \{0, 1\}^{\times 2n}$. 
Up to a column permutation, this is the Walsh-Hadamard transformation. 

Now let $\mce^{\ds{F}(2n)}$ be a FLO-twirled channel. 
Then there is a similar eigenvalue relationship where the eigenvectors are Majorana monomials.
From the parameterization in \cref{eq:mgtwirl}, we can derive the eigenvalue relation 
\begin{align}
    \mce^{\ds{F}(2n)}(\mm{\bsbet}) &= \sum_{k = 0}^{2n} \binom{2n}{k}^{-1} q_k \left(\sum_{\bsalph|k} \mm{\bsalph}^{\dagger} \mm{\bsbet} \mm{\bsalph}\right) \\
     &= \left\{\sum_{k = 0}^{2n} \left[\sum_{\bsalph | k}(-1)^{|\bsbet|k + |\bsalph \cap \bsbet|}\right] \binom{2n}{k}^{-1} q_k \right\} \mm{\bsbet}\,. \nonumber      
\end{align}
The coefficient in braces can be simplified using \cref{lem:mcoeffs} and the definitions in \cref{sec:kravchuk} to give the \emph{FLO-twirled channel eigenvalue}
\begin{align}
    \mge_{j} &= \sum_{k = 0}^{2n} (-1)^{jk}  e_k(-\mbf{1}_j \oplus \mbf{1}_{2n - j}) \binom{2n}{k}^{-1} q_k
    \label{eq:lamtopmm} \\
    &= \left[(\mbf{s} \cdot \mbf{M} \cdot \mbf{d}^{-1}) \cdot \mbf{q}\right]_j \mathrm{.}
\end{align}
This constitutes a proof of the following lemma, for which we adopt a vector notation.

\begin{lemma}(FLO-twirled Channel Eigenvalues)
\label{lem:mdef}
We have
\begin{align}
\label{eq:lamtopmmmat}
\boldsymbol{\mge} = \left(\mbf{s} \cdot \mbf{M} \cdot \mbf{d}^{-1}\right) \cdot \mbf{q} \,.
\end{align}
\end{lemma}
\begin{proof}
    Immediate. 
\end{proof}

As $2^{-\ell/2} \left[\mbf{s}^{(\ell)} \cdot \mbf{M}^{(\ell)}\right]$ is an involution, we can invert this relation and express the probabilities $\mbf{q}$ from \cref{eq:lamtopmmmat} in terms of the eigenvalues $\boldsymbol{\mge}$ as
\begin{align}
    \mbf{q} = 4^{-n} \left(\mbf{d} \cdot \mbf{M} \cdot \mbf{s} \right) \cdot \boldsymbol{\mge} \,.
    \label{eq:inverserelation}
\end{align}
Owing to this linear relation, we can learn either $\mbf{q}$ or $\boldsymbol{\xi}$ to completely characterize a FLO-twirled channel, but one must be careful about the precision in these different cases due to the relative normalization.

A FLO-twirled channel commutes with every element of the group $\ds{F}(2n)$, so we can extend the notion of a channel eigenvalue to circuits as
\begin{align}
    \widetilde{\mcu}^{\mathds{F}(2n)} \coloneqq \prod_g \left(\mcu_g \mce_g^{\ds{F}(2n)}\right) = \left(\prod_g \mcu_g\right) \left(\prod_g \mce_g^{\ds{F}(2n)}\right) \label{eq:circuitprod}.
\end{align}
This gives
\begin{align}
    \widetilde{\mcu}^{\mathds{F}(2n)}(\gamma_{\bsbet}) = \left(\prod_g \xi_{g, k} \right)\mcu(\gamma_{\bsbet}) \mbox{\hspace{5mm}} (k \coloneqq |\bsbet|)
\end{align}
where $\xi_{g, k}$ is the degree-$k$ \emph{gate eigenvalue} for the circuit factor $g$.
We define the product in parentheses as the degree-$k$ circuit eigenvalue for $\widetilde{\mcu}^{\ds{F}(2n)}$ as
\begin{align}
    \Lambda_{k} \coloneqq \prod_g \xi_{g, k} 
    \label{eq:circeigintro}
\end{align}

\section{Measurement Protocol}
\label{sec:protocol}
Efficiently estimating noise in a quantum system requires that we fit to models with certain simplifying assumptions.  
Following from \cref{def:channelnoise}, we assume the parameters describing the noise channel $\mce$ are independent of time and, furthermore, depend only on the gate we apply. 
This is similar to the time-stationary and Markovian assumptions of \cite[def.\ 2]{flammia2019efficient}, but unlike ref.~\cite{flammia2019efficient}, and much related literature, we allow wide latitude for gate dependence.
In that regard our protocol can be viewed as a kind of ``fermionic gate-set tomography'' in the FLO-twirled subspace \cite{blumekohout2013robust, merkel2013selfconsistient, kim2015microwave, nielsen2021gate}.

We broadly refer to the map that takes ideal physical operations---state preparations, unitary gate operations, and measurements---to their parameterized noisy versions as a \emph{model} (and we view the FACES model given in \cref{def:model} as a special case of this). 
Denote the set of idealized physical operations by $\mc{S}$ and the set of noisy operations by $\mc{N}$, then the model is a bijective map $\phi: \mc{S} \rightarrow \mc{N}$.
We only consider circuits composed of gates for which the model is defined.
This is the \emph{realizable setting}.
If the ``true'' noise belongs to the realizable setting, then we provide convergence/correctness guarantees that we can learn that model. 
Otherwise, one of our assumptions is violated. 
We next present some refining assumptions about our models in the following subsection.

\subsection{Model assumptions}
\label{sec:modelassumptions}

Denote a noisy quantum circuit by $\widetilde{\mcu} = \prod_{g = 1}^L \widetilde{\mcu}_g$, where each noisy gate factor is given by $\widetilde{\mcu}_g \coloneqq \mcu_g \mce_{g}$ as in \cref{def:channelnoise}, extended to a circuit of length $L$ as in \cref{eq:circuitprod}.
The ideal part $\mcu_g$ of each noisy gate is known from the physical operation we applied, but the accompanying noise channel $\mce_g$ contains free parameters to be learned.
In principle, this definition allows for the case where we have multiple ways of physically implementing the same idealized circuit, resulting in multiple circuits with the same ideal part but possibly distinct accompanying noise channels. 
As stated above, we will restrict to the setting where the applied noise channel is unique to the corresponding idealized gate (that is, our model $\phi$ is bijective), but it is straightforward to relax this assumption.
In addition, the noisy part of each gate may also depend on the idealized gate itself.
In principle, the twirling gates $\mc{U}^{\dagger}$ and $\mc{V} \mc{U} \mc{V}^{\dagger}$ in \cref{eq:lem1line1} are noisy in the FLO-twirling case as well, but we assume that this noise is only weakly correlated with $\mc{V}$, since we can implement them without necessarily implementing $\mc{V}$ due to the FLO group structure.
Our expectation is that, in practice, this noise will introduce some additional uncertainty on our gate-eigenvalue estimates.
Under these assumptions, the FACES model given in \cref{def:model} is a special case of the general model described in this section where each individual FLO unitary gate is associated to its FLO-twirled noisy version through its collection of gate eigenvalues.
In summary, we assume that any noise model we wish to learn is Markovian in at least two senses: first, that it can be described by a completely positive trace-preserving map, but also that the noise specific to a given gate depends only on the idealized gate itself.

We require that our state preparations and measurements (SPAM) are not too noisy so that we can obtain useful information from experiments.
Specifically, we define $\widetilde{\rho}_{\mbf{x}}$ to be the density matrix describing a noisy preparation of the computational basis state $\ketbra{\mbf{x}}{\mbf{x}}$.
Similarly, let $\widetilde{E}_{\mbf{x}}$ correspond to the noisy POVM element associated to measuring the computational basis state $\ketbra{\mbf{x}}{\mbf{x}}$.
Our first assumption is that SPAM errors are sufficiently weak, as quantified by the following inequality,
\begin{align}
    \bra{\mbf{0}} \widetilde{\rho}_{\mbf{0}} \ket{\mbf{0}} \bra{\mbf{x}} \widetilde{E}_{\mbf{x}} \ket{\mbf{x}} \geq 1 - c
\end{align}
for some sufficiently small constant $c$, say $c \le \frac{1}{3}$.
The noise channels $\mce$ we consider should be also such that their FLO-twirled versions are not too noisy, otherwise they will be hard to learn.
We quantify this mathematically as 
\begin{align}
||\mc{I} - \mce_j^{\ds{F}(2n)}||_{\infty} \leq c'
\end{align}
for all noise channels $\mce_j$ that we are trying to learn. 
Here again, $c'$ is a sufficiently small constant. 
Since our channels are always diagonal in the Pauli basis, this is equivalent to requiring that any noisy-channel eigenvalue lies in the interval $[1-c', 1]$ (in the proof of \cref{thm:samplecomplexity}, we require that $c' = \frac{1}{2}$). 
In practice, we routinely expect that these numbers $c, c'$ are less than $0.05$ or even $0.01$ in some quantum architectures, so these are not strong assumptions. 

\subsection{FLO-Specific Considerations}
\label{sec:floassumptions}

A fundamental challenge attendant to the setting of noisy FLO circuits arises due to the fact that the group $\ds{F}(2n)$ is continuous.
Since we can only ever estimate finitely many noise parameters, we thus assume that the noise can nonetheless be described by a finite set of channels.
For completeness, we choose a particular gate set given by $\mathcal{G} \coloneqq \{e^{i \theta Z_j}\}_{j, \theta} \cup \{G_{j}(H, H)\}_j$ in the qubit picture, where $\theta \in [0, 2\pi)$ is allowed to be an arbitrary continuous parameter for each $j$.
The matrix $H$ is the single-qubit Hadamard operator, and $G_j(A, B)$ is the 2-qubit matchgate as defined in \cref{sec:matchgates}.
While the 2-qubit gates $\{G_{j}(H, H)\}_j$ form a discrete set, the single qubit rotations $\{e^{i \theta Z_j}\}_{j, \theta}$ do not.
We partition the interval $[0, 2\pi)$ into $N$ bins of the form $B_k \coloneqq [\frac{2\pi (k - 1)}{N}, \frac{2 \pi k}{N})$ for $k \in [N]$, and we assume that the single-qubit rotation $e^{i \theta Z_j}$ is subject to a fixed noise channel $\mce_{1, jk}$ if $\theta \in B_k$.
Similarly, we assume that each of the $\{G_{j}(H, H)\}_j$ is subject to a fixed noise channel $\mce_{2, j}$
An interesting related question is whether we can use these methods to learn a finite set of \emph{noise generators} for a continuous family of noise operations.
We leave this question for future work, and we return to it in the Discussion.

Finally, we review some practical details related to twirling over the FLO group.
Importantly, FLO twirling cannot generally be made to be constant-depth. 
One way to see this is to note that, in the qubit picture, matchgate twirling averages together the error probabilities of local Paulis with those of arbitrarily high weight (e.g. $X_1$ and $Z^{\otimes (n - 1)} Y_n$). 
Thus, any FLO unitary ensemble we use to apply the twirling must involve at least one deep circuit if we require our gates to have constant support.
We again return to these points in the Discussion.

\subsection{Protocol Details}
\label{sec:protocoldetails}

Our measurement protocol for obtaining FLO-averaged circuit eigenvalues closely follows the ACES protocol of ref.~\cite{flammia2022averaged}, and it is detailed graphically in \cref{fig:controlflow}.
To learn all the eigenvalues of a FLO-twirled noise channel, we present two complementary protocols, which differ slightly in the required input state, measurement basis, and net FLO circuit applied. 
An important unitary conjugation channel in our analysis is $\mc{U}_{+}(X) = U_+ X U_+^{\dagger}$, related to the unitary $U(R_+)$ defined in \cite{helsen2022matchgate}, and satisfying
\begin{align}
\mc{U}_{+}(\ketbra{\mbf{+}}{\mbf{+}}^{\otimes n}) = \left(\frac{I + i \gamma_1}{\sqrt{2}}\right) \ketbra{\mbf{0}}{\mbf{0}} \left(\frac{I - i \gamma_1}{\sqrt{2}}\right) \mathrm{.}
\end{align}
We have the following decomposition in terms of the gate set from the previous section
\begin{align}
U_+ = e^{-i \frac{\pi}{4}Z_1} \prod_{j = 1}^{n - 1} \left[e^{-i \frac{\pi}{4}Z_{j + 1}} G_j(H, H) e^{i \frac{\pi}{4}Z_j} G_j(H, H) e^{i \frac{\pi}{4}Z_{j + 1}} \right]
\end{align}
up to an overall phase, where the product is taken in ascending order in $j$ from left to right.

In general, we label an ensemble of noisy FLO circuits as
\begin{align}
\mc{C}_r = \left\{\widetilde{\mc{U}}^{(j)} \coloneqq \prod_{m = 1}^{L(j)} \widetilde{\mc{U}}_m^{(j)}\right\}_j
\end{align}
where $r \in \{``x", ``z"\}$, and the net ideal unitary varies with $r$ as follows: 
for all $\widetilde{\mc{U}}^{(j)} \in \mc{C}_z$, we have $\mc{U}^{(j)} = \mc{I}$ (the identity channel), and for all $\widetilde{\mc{U}}^{(j)} \in \mc{C}_x$ we have $\mc{U}^{(j)} = \mc{U}_+$. 
Recall that a tilde is used to indicate the noisy version of a gate or circuit.
To avoid overloading the superscript notation, we denote $\mc{C}_r^{\mathds{F}(2n)}$ as the ensemble whose elements correspond to the FLO-twirled elements of $\mc{C}_r$, and we drop the explicit superscript on the elements of $\mc{C}_r^{\mathds{F}(2n)}$ themselves, with the understanding that they are twirled.

We can organize our circuits into two cases depending on whether the first qubit is prepared in the $x$ or $z$ basis; we call these ``type-$x$'' and ``type-$z$'', respectively. 
The readout data, and thus the expected signal, takes two different forms in these two cases.
We thus organize our empirical estimates of the outcome probabilities according to the circuit (or equivalently, readout) type.
Readout data of type-$z$ is the Hamming weight of an $n$-bit string, whereas readout data of type-$x$ consists of a sign ($\pm$) together with the Hamming weight of an $(n - 1)$-bit string. 
When we collect these data into empirical probability estimates over the relevant Hamming weights, we denote these empirical estimates by $\bfp$ regardless of the $x$/$z$ type. 
Nevertheless, we can distinguish the type by adding an index $s$ that takes values in $\{+,-,0\}$ where $0$ indicates $z$-type data, and $\pm$ corresponds to the sign label of $x$-type data. 
For example, the estimated probability of weight-$j$ in a $+$-outcome $x$-type experiment would be labeled $\bp_{+,j}$, and so on. 

We define the $z$-type Born-rule probability distribution sampled in experiment $j$ as 
\begin{align}
    \bp^{(j)}_{0, \ell} \coloneqq \sum_{\mbf{w}|\ell, n} \bra{\mbf{w}} \widetilde{\mc{U}}^{(j)}(\ketbra{\mbf{0}}{\mbf{0}}) \ket{\mbf{w}}
    \label{eq:ztypebp}
\end{align}
for $\widetilde{\mc{U}}^{(j)} \in \mc{C}_z^{\mathds{F}(2n)}$.
Similarly, define the $x$-type Born-rule probability distribution in experiment $j$ as 
\begin{align}
    \bp_{\pm, \ell}^{(j)} \coloneqq \sum_{\mbf{w}|\ell, n - 1} \bra{\pm_{y}\mbf{w}} \widetilde{\mc{U}}^{(j)}(\ketbra{\mbf{+}}{\mbf{+}}^{\otimes n}) \ket{\pm_y, \mbf{w}} \mathrm{.}
    \label{eq:xtypebp}
\end{align}
Here, the first qubit is measured in the $y$-basis, and the outcome of this measurement corresponds to the $\pm$ index on this probability.
Let $\mbf{P}^{(j)}_s$ be the vector with components $\bp^{(j)}_{s, \ell}$ for $s \in\{+, -, 0\}$

The Born-rule probability distributions are related to the circuit eigenvalues of a FLO-twirled quantum circuit.
Let $\Lambda^{(j)}_{k}$ denote the degree-$k$ eigenvalue of circuit $j$ in either ensemble $\mc{C}_x^{\mathds{F}(2n)}$ or $\mc{C}_z^{\mathds{F}(2n)}$, for $k \in [2n]$.
We collect these eigenvalues into vectors as 
\begin{align}
\boldsymbol{\Lambda}_{\mathrm{even}}^{(j)} &\coloneqq \left(\Lambda_{2k}^{(j)}\right)_{k = 0}^n \\
\boldsymbol{\Lambda}_{\mathrm{even} \setminus \{2n\}}^{(j)} &\coloneqq \left(\Lambda_{2k}^{(j)}\right)_{k = 0}^{n - 1} \\
\boldsymbol{\Lambda}_{\mathrm{odd}}^{(j)} &\coloneqq \left(\Lambda_{2k + 1}^{(j)}\right)_{k = 0}^{n - 1} ,
\end{align} 
and we have

\begin{lemma}[]
For an ensemble of noisy and FLO-twirled FLO circuits $\mc{C}^{\mathds{F}(2n)}_r$,
\begin{align}
    \mbf{P}^{(j)}_{0} = 2^{-n} \mbf{d}^{(n)} \cdot \mbf{M}^{(n)} \cdot \boldsymbol{\Lambda}_{\mathrm{even}}^{(j)}
    \label{eq:ztypeinv}
\end{align}
for $r = ``z"$. 
Similarly,
\begin{align}
    \frac{1}{2}(\mbf{P}^{(j)}_{+} + \mbf{P}^{(j)}_{-}) &= 2^{-n} \mbf{d}^{(n - 1)} \cdot \mbf{M}^{(n - 1)} \cdot \bs{\Lambda}_{\mathrm{even} \setminus \{2n\}}^{(j)}
    \label{eq:plustypeinv}\\
    \frac{1}{2}(\mbf{P}^{(j)}_+ - \mbf{P}^{(j)}_-) &= 2^{-n} \mbf{d}^{(n - 1)} \cdot \mbf{M}^{(n - 1)} \cdot \bs{\Lambda}_{\mathrm{odd}}^{(j)}
    \label{eq:minustypeinv}
\end{align}
for $r = ``x"$.
\label{lem:bptolambdatransform}
\end{lemma}
We collect these transformations into a global one over all experiments as
\begin{align}
    \bs{\Lambda} = \mbf{V} \mbf{P}
\end{align}
where 
\begin{align}
    \mbf{V} = \mbf{V}_{z} \oplus \mbf{V}_x 
\end{align}
and
\begin{align}
\begin{cases}
        \mbf{V}_z = \bigoplus_{j \in \mc{C}_z} \left[\mbf{M}^{(n)} \cdot (\mbf{d}^{(n)})^{-1}\right] & \\
         \mbf{V}_x = \bigoplus_{j \in \mc{C}_x} \left\{ \left[\mbf{M}^{(n - 1)} \cdot (\mbf{d}^{(n - 1)})^{-1}\right] \oplus \left[\mbf{M}^{(n - 1)} \cdot (\mbf{d}^{(n - 1)})^{-1}\right] \right\} \cdot (\sqrt{2} \mbf{H} \otimes \mbf{I}_n) &
\end{cases}
\label{eq:vdef}
\end{align}
with $\mbf{H}$ is the $2 \times 2$ Hadamard transform.
Using these transformations, we reconstruct the global set of circuit eigenvalues and subsequently reconstruct the the gate eigenvalues.
Our protocol proceeds as follows: given a set of elementary gates $\mc{G}$ we would like to characterize, we implement an ensemble of these circuits composed of the corresponding FLO-twirled gates.
Each circuit in our ensemble has eigenvalue given by the design matrix
\begin{align}
    \Lambda_k^{(j)} = \prod_{g \in \mc{G}} \xi_{g, k}^{A_{j, g}}
    \label{eq:designdef}
\end{align}
where the design matrix element $A_{j, g}$ corresponds to the number of occurences of gate $g$ in circuit $j$. 
Physically, we sample the Born rule probabilities defined in \cref{eq:ztypebp} and \cref{eq:xtypebp} to obtain empirical estimates $\hat{\mbf{P}}^{(j)}$ of these distributions for each circuit $j$, and we apply the transformations in \cref{lem:bptolambdatransform} to produce estimators $\hat{\mbf{\Lambda}}^{(j)}$ for the associated circuit eigenvalues.
Finally, we fit our gate eigenvalue estimates to our circuit eigenvalue data using a \emph{log-linear model}.
That is, our estimator for each Majorana degree $k$ is the solution to a system of equations of the form $\mbf{Ax} = \mbf{b}$, where $\mbf{A}$ is the design matrix, $x_{g, k} = -\mathrm{log}(\xi_{g, k})$, and $b_{j, k} = -\mathrm{log}(\Lambda^{(j)}_k)$. More formally, we give the following definition

\begin{definition}[Estimator]
    For a FACES model $\mc{M} = (\mc{G}, \bs{\xi}, \mc{C})$, from which we obtain the design matrix $\mbf{A}$, the FACES estimator for $\hat{\bs{\xi}}$ is taken as
    $\hat{\xi}_{g, k} = \exp(-\hat{x}_{g, k})$ for each $g$, $k$, with
\begin{equation}
\hat{x}_{g, k} = \begin{cases}
\left(\mbf{A}^+ \hat{\mbf{b}}\right)_{g, k} & \left(\mbf{A}^+ \hat{\mbf{b}}\right)_{g, k} \geq 0 \\
0 & \mathrm{otherwise}
\end{cases}
\end{equation}
\label{def:estimator}
\end{definition}
That is, we take the estimator given by ordinary least-squares, projected onto the positive orthant for each Majorana degree $k$. 
One could of course consider other estimators, such as $\hat{\mbf{x}}'$ taken to be the element of the positive orthant that minimizes $\lVert \mbf{A}\hat{\mbf{x}}' - \hat{\mbf{b}} \rVert_2$, or even an estimator that maximizes a likelihood function subject to a positivity constraint. 
These are convex programs and would presumably be efficient to consider as practical estimators. 
However, our choice of estimator is convenient from the standpoint of proving sample-complexity guarantees. 
We give a pseudocode description of the relevant subroutines in \cref{sec:pseudocode}.

\subsection{Sample Complexity}
\label{sec:complexity}
In this section, we analyze the sample complexity required to bound the error on the estimator $\hat{\bs{\xi}}$ in \cref{def:estimator}.
We provide the following theorem:
\begin{theorem}
   Suppose we have both $\Lambda_k^{(j)} \geq \frac{1}{2}$ and $\hat{\Lambda}_{k}^{(j)} \geq \frac{1}{4}$ for all degree sectors $k$ and experiments $j$, then the FACES estimator satisfies
    \begin{align}
    \lVert \hat{\bs{\xi}} - \bs{\xi} \rVert_{\infty} \leq 4 \lVert \mbf{A}^+ \rVert_{\infty} \varepsilon
    \end{align}
    We thus require $S = O\left(\frac{n \lVert \mbf{A}^+ \rVert_{\infty}^2\log(J/\delta)}{\epsilon^2}\right)$ samples per experiment to guarantee error $\epsilon$ with probability $1 - \delta$.
    \label{thm:samplecomplexity}
\end{theorem}

We prove this in two parts. 
We first prove a lemma that bounds the sensitivity of the gate-noise eigenvalue estimates in terms of the sample estimates of the Born rule probabilities of each experiment. 
The main theorem on sample complexity then follows from well-known results about estimating probability distributions from samples in the 1-norm.

\begin{lemma}
    \label{lem:samplecomp}
    Suppose we have both $\Lambda_k^{(j)} \geq \frac{1}{2}$ and $\hat{\Lambda}_{k}^{(j)} \geq \frac{1}{4}$ for all degree sectors $k$ and experiments $j$, then for all $\varepsilon \in (0, 1/4]$ we have
    \begin{align}
    \lVert \hat{\bs{\xi}} - \bs{\xi} \rVert_{\infty} \leq 4 \lVert \mbf{A}^+ \rVert_{\infty} \varepsilon 
    \end{align}
    where 
    \begin{align}
        \mathrm{max}_j \lVert \delta \mbf{P}^{(j)}\rVert_1  \le \varepsilon
    \end{align}
    is a uniform upper bound on the 1-norm error in the estimate of the Born rule probabilities.
\end{lemma}

Here the matrix norm $\|\cdot\|_\infty$ is defined in terms of the vector $\infty$-norm as 
\begin{align}
    \|\mbf{A}\|_\infty = \max_{x\neq 0} \frac{\|\mbf{Ax}\|_\infty}{\|\mbf{x}  \|_\infty}
\end{align}
and can be computed as the maximum value of the vector $1$-norm of each row. 

Our proof relies on several applications of the mean value theorem, which we recall is stated as:
\begin{lemma}[Mean-Value Theorem]
\label{thm:mvthm}
Let $f: [a, b] \rightarrow \mathds{R}$ be a continuous function on the closed interval $[a, b]$ and differentiable on the open interval $(a, b)$ for $a < b$.
There exists some $c$ in $(a, b)$ such that
\begin{align}
    f'(c) = \frac{f(b) - f(a)}{b - a}\,.
\end{align}
\end{lemma}

\begin{proof}{(\cref{lem:samplecomp})}
The logarithms of the circuit and gate eigenvalues satisfy a linear relationship
\begin{align}
    \mbf{A}\mbf{x} = \mbf{b}
\end{align}
where $x_i = -\mathrm{log}(\xi_i)$, and $b_i = -\mathrm{log}(\Lambda_i)$, and we have bundled the compound indices into single indices for simplicity.
Letting $\delta x_i  \equiv \hat{x}_i - x_i$, define the function
\begin{align}
u_i(t) = \mathrm{exp}\left(x_i + t \delta x_i \right)
\end{align}
for each $i$, such that $u_i(0) = \xi_i$ and $u_i(1) = \hat{\xi}_i$.
Applying \cref{thm:mvthm} to each such function separately,
\begin{align}
    \hat{\xi}_i - \xi_i &= u_i(1) - u_i(0) \\
    &= u'_i(c_i) \\
    \hat{\xi}_i - \xi_i &= \delta x_i \mathrm{exp}\left(x_i + c_i \delta x_i \right)
\end{align}
for some $c_i \in (0, 1)$
The exponential function is monotonic, so we are guaranteed that the absolute value of the error is less than or equal to the value of the right-hand side for $c_i = 0$ (if $\delta x_i < 0$) or $c_i = 1$ (if $\delta x_i > 0$).
Using the definition of the gate eigenvalue and the estimate thereof gives
\begin{align}
    |\hat{\xi}_i - \xi_i| &\leq |\delta x_i| \mathrm{max}(\hat{\xi}_i, \xi_i) \leq |\delta x_i|
\end{align}
where the second inequality follows since we are assuming the true gate eigenvalue and its estimate are upper-bounded by 1.
This gives the inequality in terms of vector norms
\begin{align}
    \lVert \hat{\bs{\xi}} - \bs{\xi} \rVert_{\infty} \leq \lVert \delta \mbf{x} \rVert_{\infty}.
    \label{eq:xixbound}
\end{align}
Next, let $\tilde{\mbf{x}} \coloneqq \mbf{A}^+ \hat{\mbf{b}}$ be the ordinary least-squares estimator.
We have
\begin{align}
    |\delta x_i| = \begin{cases}
        |\tilde{x}_i - x_i| & \tilde{x}_i > 0 \\
        |x_i| & \tilde{x}_i \leq 0  
    \end{cases}
\end{align}
for all $i$.
If $\tilde{x_i} \leq 0$, we have
\begin{align}
    |x_i| \leq |x_i + |\tilde{x}_i|| = |x_i - \tilde{x}_i| = |\tilde{x}_i - x_i|
\end{align}
Thus $|\delta x_i| \leq |\tilde{x}_i - x_i|$ for all $i$, and this leads to $\lVert \delta \mbf{x}\rVert_{\infty} \leq \lVert \tilde{\mbf{x}} - \mbf{x}\rVert_{\infty}$. 
By the assumptions that we are in the realizable setting and that $\mbf{A}$ is full-rank, we have $\mbf{x} = \mbf{A}^+ \mbf{b}$, and this gives
\begin{align}
\lVert \hat{\bs{\xi}} - \bs{\xi} \rVert_{\infty} \leq \lVert \mbf{A}^+ (\hat{\mbf{b}} - \mbf{b}) \rVert_{\infty} \leq \lVert \mbf{A}^+ \rVert_{\infty} \lVert \hat{\mbf{b}} - \mbf{b} \rVert_{\infty}
\label{eq:xibound}
\end{align}
where the second inequality follows from submultiplicativity of the norm.

We bound the error on $\mbf{b}$ similarly. 
Letting $\delta\Lambda_i = \hat{\Lambda}_i - \Lambda_i$, we define
\begin{align}
    v_i(t) = \mathrm{log}(\Lambda_i + t \delta \Lambda_i)
\end{align}
for each index $i$ such that $v_i(0) = b_i$ and $v_i(1) = \hat{b}_i$.
We have
\begin{align}
    \hat{b}_i - b_i &= v_i(1) - v_i(0) \\
    &= v'_i(\tilde{c}_i) \\
    \hat{b}_i - b_i &= \frac{\delta \Lambda_i}{\Lambda_i + \tilde{c}_i \delta \Lambda_i} 
\end{align}
for $\tilde{c}_i \in (0, 1)$.
The function on the right-hand side is monotonic in $\tilde{c}_i$, so its absolute value is maximized for $\tilde{c}_i = 0$ (if $\delta \Lambda_i > 0$) or $\tilde{c}_i = 1$ (if $\delta \Lambda_i < 0$).
This gives
\begin{align}
| \hat{b}_i - b_i | &\leq |\delta \Lambda_i | \mathrm{max} \left[\hat{\Lambda}_i^{-1}, \Lambda_i^{-1}\right] \\
| \hat{b}_i - b_i | &\leq 4|\delta \Lambda_i | 
\end{align}
by our assumption that both $\Lambda_i \geq \frac{1}{2}$ and $\hat{\Lambda}_i \geq \frac{1}{4}$ for all $i$.
Finally, we use this to provide a bound on $\lVert \hat{\mbf{b}} - \mbf{b} \rVert_{\infty}$ as 
\begin{align}
    \lVert \hat{\mbf{b}} - \mbf{b} \rVert_{\infty} &\leq 4 \lVert \delta \boldsymbol{\Lambda} \rVert_{\infty} \\  
    &= 4 \lVert \mbf{V} \delta \mbf{P} \rVert_{\infty} \\
    &\leq \mathrm{max}_j \left(4 \lVert \delta \mbf{P}^{(j)} \rVert_{1} \right)\\
    \lVert \hat{\mbf{b}} - \mbf{b} \rVert_{\infty} &\leq 4 \varepsilon \label{eq:bbound}.
\end{align}
From the second to the third line above, we replaced the maximum over indices $i$ in $\lVert \cdot \rVert_{\infty}$ with a maximization first over each experiment and then within each experiment. 
We then applied the norm inequality $\lVert \cdot \rVert_{\infty} \leq \lVert \cdot \rVert_{1 \rightarrow \infty} \lVert \cdot \rVert_{1}$ experiment by experiment and used the fact that the maximum absolute element of $\mbf{V}$ is 1 in each case from \cref{eq:vdef}.
Combining \cref{eq:bbound} with \cref{eq:xibound} gives the result.
\end{proof}

\section{Numerics}
\label{sec:numerics}

\begin{figure}
    \centering
    \includegraphics[width=\textwidth]{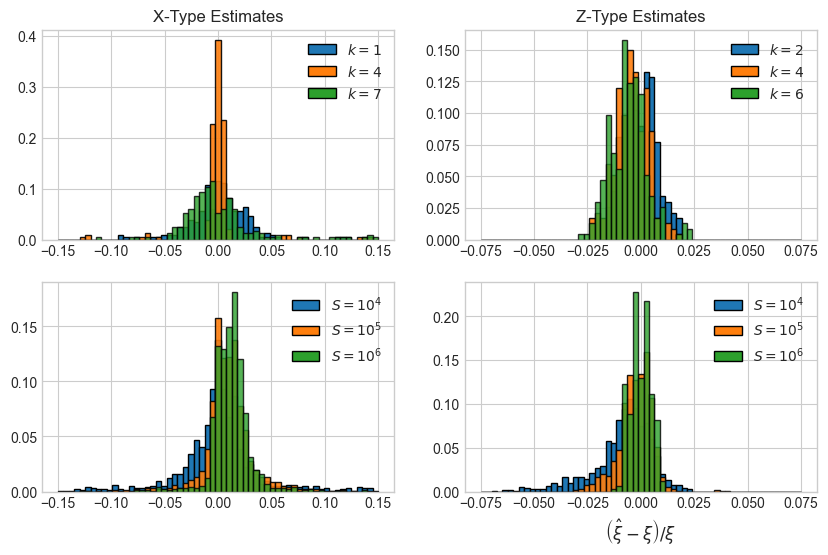}
    \caption{
    Histogram of relative errors for $2n = 10$ fermionic modes with $1000$ circuits for each experiment type. See main text for description.
    }
    \label{fig:hist}
\end{figure}

Finally, we demonstrate the efficacy of our protocol numerically in \cref{fig:hist}, where we plot a histogram of relative gate-eigenvalue errors for $n=5$ qubits ($2n = 10$ fermionic modes) and $m=1000$ circuits for each experiment type.
Our noise model is given by a random $2$-qubit Pauli channel following each gate we apply, where the Pauli channel probabilities are chosen uniformly from the interval $[10^{-2} - 10^{-3}, 10^{-2} + 10^{-3}]$.
Columns are partitioned according to $x$- and $z$-type estimates, and rows are partitioned according to whether the different histograms correspond to different Majorana weight-$k$ or different numbers $S$ of samples from the Born-rule probability distribution (we assume that we can apply FLO-twirled circuits exactly).
We also partition $Z$-rotation angles into 46 bins of width $\Delta \theta = \frac{\pi}{23}$.
We observe that the protocol is able to routinely achieve less than 5\% relative error for most of the eigenvalues estimated, and that this steadily improves with the number of shots.

\section{Conclusion}
\label{sec:conclusion}
We have shown how to learn the matchgate-averaged noise parameters on many matchgates simultaneously via FACES, our fermionic adaptation of the ACES protocol.
We envision that this will enable experiments to build a complete characterization of averaged noise parameters for universal gate sets, though several open problems remain.
Foremost is the technical challenge of performing matchgate twirling.
Unlike Pauli twirling, it seems unlikely that matchgate twirling can be achieved in constant depth.
However, there have recently been promising results to show that low-depth circuits can nevertheless approximate Haar random unitaries \cite{schuster2024random, laracuente2024approximate}.
If a corresponding result holds for Haar random matchgates, then this issue can be mitigated, if not entirely avoided.
Alternatively, we might consider only reproducing the Haar ensemble of FLO circuits up to the second moment, which is sufficient for exactly twirling the channel.
It is known that this can be achieved using only gates in the intersection of the matchgate and Clifford groups \cite{wan2023matchgate, heyraud2024unified}.
Finally, we might employ restricted circuit structure to address this challenge in a similar strategy to interleaved randomized benchmarking (IRB) \cite{magesan2012efficient}, which has historically enjoyed a great deal of  success in practice.
In a setting where the circuit structure is limited to the FLO analog of IRB, for example, we expect that FACES will exhibit a ``repeated convolution structure'' similar to that of matchgate benchmarking \cite{helsen2022matchgate}, since we are considering Haar averages over the group of FLO unitaries.

Another open question involves the notion of a \emph{gauge} freedom in the learned parameters.
Formally, such a freedom exists when the design matrix has a nontrivial kernel, whose elements correspond to unlearnable functions of gate eigenvalues.
Recent works~\cite{chen2023learnability,chen2024efficient} have shown that these functions correspond to the cut space of a particular graph in the Clifford setting.
It is thus germane to ask whether a corresponding characterization is possible in the FLO setting. 

It was also shown in~\cite{chen2024efficient} that, in the Clifford setting, one can construct design matrices that enable learning all learnable Pauli eigenvalues to \textit{relative} precision. 
We expect similar results to hold in the setting of FACES, but we leave this open. 

Ultimately, we would like to utilize these protocols to not only learn circuit-level noise, but to also mitigate it.
With this in mind, an important feature of our protocol is that the design matrix is necessarily block diagonal.
In the aforementioned question of gauge parameters, it is already clear that each unlearnable function thus depends only on gate eigenvalues from the same Majorana-degree sector. 
Additionally, noisy expectation values of fixed Majorana-degree observables are necessarily linear functions of the ideal expectation values, even in the setting where these observables are not Majorana monomials (e.g. measuring the expected energy for a Hamiltonian of fixed Majorana degree).
This contrasts with the setting of Pauli observables under Clifford-circuit evolution, and is reminiscent of training methods for error mitigation \cite{montanaro2021error, papic2025nearterm}.
We leave a careful consideration of this problem and other applications of our results to quantum simulation for future work.

\begin{acknowledgments} 
We are grateful to Ashley Montanaro and Ryan O'Donnell for informing us that $\mbf{M}^{(\ell)}$ was known in the literature as a Kravchuk matrix.
Matchgate circuit simulations were performed using the GitHub repository \cite{ClaesGithub}.
\end{acknowledgments}

\appendix
\section{Pseudocode Algorithms}
\label{sec:pseudocode}
In this section we provide pseudocode implementations for the steps of our protocol described in \cref{fig:controlflow}.
Algorithm \ref{alg:maces} is the main function, which learns the gate eigenvalues $\boldsymbol{\xi}$ describing an error model $\mc{M}$ for a collection of circuits $\mc{C}$ in a given number of measurements per circuit $S$.
First, the design matrix $\mbf{A}$ is computed using \cref{alg:design}.
The rows of this matrix are indexed by circuits in $\mc{C}$, and the columns are indexed by noise parameters, such that $A_{jk}$ corresponds to the number of times the noise parameter $k$ appears in circuit $j$.
In \cref{alg:SampleCircuit}, a probability distribution of outcomes is generated for each circuit according to its type (either ``$x$'' or ``$z$'').
These distributions are then input to \cref{alg:CircuitEigs}, where they are used to calculate the circuit eigenvalues by \cref{eq:ztypeinv} or eq.s (\ref{eq:plustypeinv}) and (\ref{eq:minustypeinv}).
Finally, this is used to produce the gate eigenvalues $\boldsymbol{\xi}$ using a similar linear inversion procedure to the ACES protocol.
In this procedure, $\mbf{A}$ is truncated to only include its elements which lie above some given cutoff $\delta$.
If $\mbf{A}$ is still full rank, we solve the linear equation $\mbf{A} \mbf{x} = \mbf{b}$ relating the logarithms of the circuit, $\mbf{b}$, and gate eigenvalues $\mbf{x}$. 
Finally, we recover the gate eigenvalues themselves by exponentiating the components of $\mbf{x}$. 

\begin{algorithm}[H]
    \renewcommand{\baselinestretch}{1}
    \caption{\label{alg:maces}
        $\textsc{FACES}(\mc{M},S)$\\
        Given a FACES model $\mc{M} \coloneqq (\mc{G}, \bs{\xi}, \mc{C})$ and a number of samples $S$ per circuit, output estimates $\hat{\boldsymbol{\xi}}$ of all the gate eigenvalues $\boldsymbol{\xi}$. 
        \smallskip
    }
    \begin{algorithmic}[1]
        \Require $\mc{M},S$
        \State \# Compute the design matrix $\mbf{A}$
        \State $\mbf{A} \gets \textsc{Design}(\mc{M})$
        \State $\boldsymbol{\Lambda} \gets [\ \ ]$ \Comment{Initialize an empty list}
        \State \# estimate all circuit eigenvalues
        \For{$\mc{U}^{(j)} \in \mc{C}$}
            \State $\boldsymbol{p}^{(j)} \gets \textsc{SampleCircuit}(\mc{U}^{(j)},S)$
            \State $\boldsymbol{\Lambda}^{(j)} \gets \textsc{CircuitEigs}(\boldsymbol{p}^{(j)})$
            \State $\boldsymbol{\Lambda} \gets \textsc{append}(\boldsymbol{\Lambda}, \boldsymbol{\Lambda}^{(j)})$
        \EndFor\\
        \# Estimate the gate eigenvalues
        \State $\hat{\boldsymbol{\xi}} \gets \textsc{GateEigs}({\boldsymbol{\Lambda}}, \mbf{A})$
        
        \Ensure $\hat{\boldsymbol{\xi}}$
    \end{algorithmic}
\end{algorithm}

\begin{algorithm}[H]
    \renewcommand{\baselinestretch}{1}
    \caption{\label{alg:design}
        $\textsc{Design}(\mc{M})$\\
        Compute the design matrix $\mbf{A}$ for the FACES model $\mc{M} \coloneqq (\mc{G}, \bs{\xi}, \mc{C})$.
        \smallskip
    }
    \begin{algorithmic}[1]
        \Require $\mc{M}$
        \State $J \gets |\mc{C}|$
        \State $K \gets |\mc{G}|$
        \State $\mbf{A} \gets \mbf{0}_{J \times K}$ \Comment{Initialize $\mbf{A}$ as an all-zeros matrix.} \\
        \# Compute the design matrix $\mbf{A}$
        \For{$\mcu^{(j)} \coloneqq \prod_{m = 1}^{L(j)} \mcu^{(j)}_m \in \mc{C}$}
        \For{$m \in [L(j)]$} 
        \For{$\ell \in [K]$}
        \If{$\mcu^{(j)}_m = G_{\ell}$}
        \State $A_{j \ell}$++
        \EndIf
        \EndFor
        \EndFor 
        \EndFor
        \, \# Check if $A$ has full column-rank.
        \If{$\textsc{rank}(A) < K$}
            \State $\textsc{throw}$(``$\mbf{A}$ is not full rank.
            \State \hspace{3.5em} $\mc{M}$ is not identifiable from $\mc{C}$.'')
        \EndIf
        \Ensure $\mbf{A}$
    \end{algorithmic}
\end{algorithm}

\begin{algorithm}[H]
    \renewcommand{\baselinestretch}{1}
    \caption{\label{alg:SampleCircuit}
        $\textsc{SampleCircuit}(\mc{U}, S)$\\ 
        Draw $N$ iid samples from the twirled ensemble for circuit $\mc{U}$ and return an estimate $\boldsymbol{p}$ of the measurement outcome probability distribution. 
        \smallskip
    }
	\begin{algorithmic}[1]
		\Require $\mc{U}, S$ \\
        \# Set $\boldsymbol{p}$ to all zeros (for circuit type $x$ or $z$). 
        \State $\boldsymbol{p} \gets 0$
        \For{$n \in [N]$}
		  \State $\mc{U}' \gets \textsc{Twirl}(\mc{U})$ \Comment{twirled version of $\mc{U}$}
          \State \# Run experiment, record single outcome. 
          \State $w \gets \textsc{RunExperiment}(\mc{U}')$
          \State $\boldsymbol{p}[w]$++ \Comment{Increment count for outcome $w$}
		\EndFor
		\Ensure $\boldsymbol{p}/S$
	\end{algorithmic}
\end{algorithm}

\begin{algorithm}[H]
    \renewcommand{\baselinestretch}{1}
	\caption{\label{alg:CircuitEigs}
        $\textsc{CircuitEigs}(\boldsymbol{p})$\\ 
        Given a probability distribution $\boldsymbol{p}$, compute the associated circuit eigenvalues $\boldsymbol{\Lambda}$ depending on the type of $\boldsymbol{p}$.
        \smallskip
    }
	\begin{algorithmic}[1]
		\Require $\boldsymbol{p}$ \\
        \# $\mbf{M}$ is calculated using \cref{lem:mcoeffs}. 
		\If{\textsc{type}($\boldsymbol{p}$) = $z$} \\
        \# Apply \cref{eq:ztypeinv}
          \State $\Lambda_{2k} \gets  \sum_{\ell = 0}^n \binom{n}{\ell}^{-1} M^{(n)}_{k \ell} p_{\ell}$
        \Else \Comment{else \textsc{type}($\boldsymbol{p}$) = $x$} \\
        \# Apply \cref{eq:plustypeinv} and \cref{eq:minustypeinv}
		  \State $\Lambda_{2k} \gets  \sum_{\ell = 0}^{n-1} \binom{n - 1}{\ell}^{-1} M^{(n - 1)}_{k \ell} (p_{+, \ell} + p_{-, \ell})$
          \State $\Lambda_{2k + 1} \gets  \sum_{\ell = 0}^{n-1} \binom{n - 1}{\ell}^{-1} M^{(n - 1)}_{k \ell} (p_{+, \ell} - p_{-, \ell})$
		\EndIf
		\Ensure $\boldsymbol{\Lambda}$
	\end{algorithmic}
\end{algorithm}

\begin{algorithm}[H]
\caption{$\textsc{Twirl}(\mc{U} \coloneqq \prod_{m = 1}^L \mc{U}_{m})$}\label{alg:mtwirl}
\begin{algorithmic}[1]
\Require{$\mc{U}$}
\State $\mc{U}' \gets \mc{I}$
\For{$k \in [L]$}
\State $\mcv \gets \textsc{sample}(\ds{F}(2n))$
\State $\mc{U}' \gets \left(\mc{U}_{k} \circ \mcv \circ \mc{U}_k^{\dagger}\right) \circ \left(\mc{U}_{k} \circ \mcv^{\dagger}\right) \circ \mc{U}'$

\EndFor
\Ensure $\mc{U}'$
\end{algorithmic}
\end{algorithm}

\begin{algorithm}[H]
    \renewcommand{\baselinestretch}{1}
	\caption{\label{alg:GateEigs}
        $\textsc{GateEigs}\bigl(\mbf{A}, \hat{\boldsymbol{\Lambda}}, \delta \bigr)$\\ 
        From a design matrix $\mbf{A}$, a collection of circuit eigenvalue estimates $\hat{\boldsymbol{\Lambda}}$, and a cutoff parameter $\delta$, return estimates $\hat{\boldsymbol{\xi}}$ of the gate eigenvalues $\boldsymbol{\xi}$.
        \smallskip
    }
	\begin{algorithmic}[1]
		\Require $\mbf{A}, \hat{\boldsymbol{\Lambda}}, \delta$
        \State $J, K \gets \textsc{size}(\mbf{A})$
        \State \# Delete $\hat{\bs{\Lambda}}$ and associated row of $\mbf{A}$ if below the cutoff.
        \State $\hat{\boldsymbol{\Lambda}} \gets \hat{\boldsymbol{\Lambda}}[\hat{\boldsymbol{\Lambda}} > \delta]$ 
        \State $\mbf{A} \gets \mbf{A}[\hat{\boldsymbol{\Lambda}} > \delta, \,:\,]$ 
        \State \# Validate that $\mbf{A}$ is still full column rank.
        \If{$\textsc{rank}(A) < K$}
            \State $\textsc{throw}$(``$\mbf{A}$ is not full rank with cutoff $\delta$.'')
        \EndIf
        
        \State $\hat{\mbf{b}} \gets -\log(\hat{\boldsymbol{\Lambda}})$
        \State $\hat{\mbf{x}} \gets \textsc{solve}(\mbf{A},\hat{\mbf{b}})$ \Comment{Solve $\mbf{A} \hat{\mbf{x}}=\hat{\mbf{b}}$}
        \State $\hat{\boldsymbol{\xi}} \gets \exp(-\hat{\mbf{x}})$ 
		\Ensure $\hat{\boldsymbol{\xi}}$
	\end{algorithmic}
\end{algorithm}

\section{Proofs}
\label{sec:proofs}

\subsection{Proof of \texorpdfstring{\cref{lem:mcoeffs}}{Lemma 1}} 
\label{sec:erelations}

\begin{replemma}{lem:mcoeffs}[Restatement]
We have the alternate expressions for $M^{(\ell)}_{jk}$
\begin{align}
    M_{jk}^{(\ell)} &= \mathrm{coeff}_k\bigl((1 - u)^j (1 + u)^{\ell - j} \bigr)  \label{eq:polyexpansionrep} \\
    &= \sum_{\bsalph|k, \ell} (-1)^{|\bsalph \cap [j]|} \label{eq:binsumrep} \\
    M_{jk}^{(\ell)} &= \sum_{m = \max{(0, k + j - \ell)}}^{\min{(j, k)}} \binom{\ell - j}{k - m} \binom{j}{m} (-1)^m\,.
    \label{eq:binomexpansionrep}
\end{align}
Additionally, $\mathbf{M}^{(\ell)}$ is proportional to an involution
\begin{align}
    (\mathbf{M}^{(\ell)})^2 = 2^{\ell} \mathbf{I} \mathrm{.}
    \label{eq:involutionrep}
\end{align}
\end{replemma}

\begin{proof}
    \cref{eq:polyexpansionrep} immediately follows from the definition \cref{eq:mdef} and Vieta's formulas.

    Now consider \cref{eq:binsumrep}, and let $\boldsymbol{b} \coloneqq -\mbf{1}_j \oplus \mbf{1}_{\ell - j}$ from the definition in \cref{eq:mdef}. 
    We have
    \begin{align}
    e_k(-\mbf{1}_j \oplus \mbf{1}_{\ell - j})
        &= \sum_{\bsalph|k, \ell} \prod_{g \in \bsalph} b_{g}    \\
        &= \sum_{\bsalph|k, \ell} \left(\prod_{g \in \bsalph \cap [j]} b_{g}\right)  \left(\prod_{g \in \bsalph \setminus [j]} b_{g}\right) \\
    e_k(-\mbf{1}_j \oplus \mbf{1}_{\ell - j}) &= \sum_{\bsalph|k, \ell} (-1)^{|\bsalph \cap [j]|} \mathrm{,}
    \end{align}
    yielding the result \cref{eq:binsumrep}. 

    To show \cref{eq:binomexpansionrep}, we expand the polynomial in the right side of \cref{eq:polyexpansionrep} by the binomial theorem
    \begin{align}
        (1 + u)^{\ell - j} (1 - u)^j &=\left[\sum_{m = 0}^{\ell - j}\binom{\ell - j}{m} u^m \right]\left[\sum_{m = 0}^{j}\binom{j}{m} (-u)^m \right] \\
        (1 + u)^{\ell - j} (1 - u)^j &= \sum_{k = 0}^{\ell} \left[\sum_{m = \mathrm{max}(0, k + j - \ell)}^{\mathrm{min}(j, k)} \binom{\ell - j}{k - m} \binom{j}{m} (-1)^m \right] u^k
        \label{eq:polyexpandwu}
    \end{align}
    The limits of the sum are such that the binomial coefficients are well-defined (this is redundant since our convention is such that the binomial coefficients are zero if the arguments are outside of the bounds).
    The lower limit is such that $m \geq 0$ to satisfy the lower bound of $\binom{j}{m}$ and $k - m \leq \ell - j$ to satisfy the upper bound of $\binom{\ell - j}{k - m}$.
    The upper bound is such that $m \leq j$ to satisfy the upper bound of $\binom{j}{m}$ and $m \leq k$ to satisfy the lower bound of $\binom{\ell - j}{k - m}$. 
    \Cref{eq:polyexpandwu}, together with \cref{eq:polyexpansionrep}, gives the result.

    To show \cref{eq:involutionrep}, let $v_j(u) \equiv u^j$ for $j \in \{0 \} \cup [\ell]$.
    We have, again by \cref{eq:mdef} and Vieta's formulas,
    \begin{align}
    \left[\mbf{M}^{(\ell)} \cdot \boldsymbol{v}(u)\right]_j &= (1 + u)^{\ell - j}(1 - u)^{j} \label{eq:coeffline} \\
    \mbf{M}^{(\ell)} \cdot \boldsymbol{v}(u) &= (1 + u)^{\ell} \boldsymbol{v}\left(\frac{1 - u}{1 + u}\right)
    \label{eq:hadamardmobius}
    \end{align}
    This gives
    \begin{align}
        (\mbf{M}^{(\ell)})^2 \cdot \boldsymbol{v}(u) &= (1 + u)^{\ell} \mbf{M}^{(\ell)} \cdot \boldsymbol{v}\left(\frac{1 - u}{1 + u}\right) \\
        &= (1 + u)^{\ell} \left[1 + \left(\frac{1 - u}{1 + u}\right)\right]^{\ell} \boldsymbol{v}\left[\frac{1 - \left(\frac{1 - u}{1 + u}\right)}{1 + \left(\frac{1 - u}{1 + u}\right)}\right] \\
        &= \left[(1 + u) + (1 - u)\right]^{\ell} \boldsymbol{v}\left[\frac{(1 + u) - (1 - u)}{(1 + u) + (1 - u)}\right] \\
        (\mbf{M}^{(\ell)})^2 \cdot \boldsymbol{v}(u) &= 2^{\ell} \boldsymbol{v}(u).
    \end{align}
    The set of vectors $\{\boldsymbol{v}(u)\}_{u \in \mathds{C}}$ form a spanning set for $\mathds{C}^{\ell + 1}$, so we must therefore have
    \begin{align}
        (\mbf{M}^{(\ell)})^2 = 2^{\ell} \mbf{I} \mathrm{.}
    \end{align}
\end{proof}

\subsection{Proof of \texorpdfstring{\cref{lem:karelations}}{Lemma 6}}
\label{sec:kaproof}
    
\begin{replemma}{lem:karelations}[Restatement]
The Kravchuk matrix and antipode involution commute, and their product has components
\begin{align}
    (\mbf{s}^{(\ell)} \cdot \mbf{M}^{(\ell)})_{jk} = (\mbf{M}^{(\ell)} \cdot \mbf{s}^{(\ell)})_{jk} = (-1)^{jk} e_k(-\mbf{1}_j \oplus \mbf{1}_{\ell - j}) \mathrm{.}
    \label{eq:permmdefrep}
\end{align}
\end{replemma}
\begin{proof}
    We show that both matrix orderings have the components given by the expression on the right-hand side.
    First, note
    \begin{align}
        (\mbf{s}^{(\ell)} \cdot \mbf{M}^{(\ell)})_{jk} = \begin{cases}
            \mathrm{M}^{(\ell)}_{j, k} & j \ \mathrm{even} \\
            \mathrm{M}^{(\ell)}_{\ell - j, k} & j \ \mathrm{odd}
        \end{cases} .
    \end{align}
Since exchanging rows $j$ and $\ell - j$ in $\mbf{M}^{(\ell)}$ is equivalent to multiplying the argument of the symmetric polynomial in \cref{eq:mdef} by a factor of $-1$, we have
    \begin{align}
        (\mbf{s}^{(\ell)} \cdot \mbf{M}^{(\ell)})_{jk} &= e_k[(-1)^j(-\mbf{1}_j \oplus \mbf{1}_{\ell - j})] \\
        (\mbf{s}^{(\ell)} \cdot \mbf{M}^{(\ell)})_{jk} &= (-1)^{jk}e_k(-\mbf{1}_j \oplus \mbf{1}_{\ell - j}) \label{eq:leftorderline2}
    \end{align}
where the second line follows from the identity $e_k(-\boldsymbol{v}) = (-1)^k e_k(\boldsymbol{v})$.
To show the result for the other ordering, we utilize the \emph{reciprocal polynomial} of $p(x)$, given by $p^*(x) = x^d p(1/x)$, where $d$ is the degree of $p$.
In particular, the degree-$k$ coefficient of $p$ is the degree-$(d-k)$ coefficient of $p^*$.
For any $k$, we have by \cref{lem:mcoeffs}
\begin{align}
\mathrm{M}^{(\ell)}_{j, \ell - k} &= \mathrm{coeff}_{\ell - k}\bigl((1 - u)^j (1 + u)^{\ell - j} \bigr) \\
&= \mathrm{coeff}_{k}\bigl(u^{\ell}\bigl(1 - \frac{1}{u}\bigr)^j \bigl(1 + \frac{1}{u} \bigr)^{\ell - j} \bigr) \\
&= \mathrm{coeff}_{k}\bigl((u - 1)^j (u + 1)^{\ell - j} \bigr) \\
\mathrm{M}^{(\ell)}_{j, \ell - k} &= (-1)^j \mathrm{M}^{(\ell)}_{j, k}
\end{align}
Since $\mbf{s}^{(\ell)}$ only exchanges columns $k$ and $\ell - k$ in $\mbf{M}^{(\ell)}$ if $k$ is odd, we have
\begin{align}
    (\mbf{M}^{(\ell)} \cdot \mbf{s}^{(\ell)})_{jk} = (-1)^{jk} \mathrm{M}_{jk} = (-1)^{jk}  e_k(-\mbf{1}_j \oplus \mbf{1}_{\ell - j}) \label{eq:rightorder}
\end{align}
Together, \cref{eq:leftorderline2} and \cref{eq:rightorder} complete the proof.
\end{proof}

\subsection{Proof of \texorpdfstring{\cref{lem:ptwirl}}{Lemma 7}}
\label{sec:twirlrelations}

\begin{replemma}{lem:ptwirl}[Restatement]
Consider a Pauli channel $\mce^{\ds{P}(n)}$ with Pauli error probabilities $p_{\mathbf{x}}$. 
Let $\mbx(\bsalph)$ denote the Pauli string associated to the subset of Majorana modes $\bsalph$ by the Jordan-Wigner transformation. 
Then $\mce^{\ds{F}(2n)} = (\mce^{\ds{P}(n)})^{\ds{F}(2n)}$, and the fermionic error probabilities $q_k$ of $\mce^{\ds{F}(2n)}$ are given by
\begin{align}    
q_k = \sum_{\bsalph | k} p_{\mbx(\bsalph)} \,. \label{eq:ptomgrenormrep}
\end{align}
\end{replemma}

\begin{proof}

The first statement follows from the fact that every Pauli operator is also a FLO unitary by the Jordan-Wigner transformation.
Namely, we have, for all $\bsalph \subseteq [2n]$ with $|\bsalph| = L$,
\begin{align}
    \mm{\bsalph} = \begin{cases}
    \prod_{j = 1}^{L/2} \left(\mm{\alpha_{2j-1}}\mm{\alpha_{2j}}\right) & L \ \mathrm{even} \\
    \mm{1} \left(\mm{1}\mm{\alpha_1}\right)\prod_{j = 1}^{(L - 1)/2} \left(\mm{\alpha_{2j}}\mm{\alpha_{2j+ 1}}\right) & L \ \mathrm{odd} \\
    \end{cases} \mathrm{.}
    \label{eq:mgrep}
\end{align}
Each factor in parentheses is a FLO transformation, since $\mm{j}\mm{k} = \exp\left(\frac{\pi}{2}\mm{j}\mm{k}\right)$, and $\mm1 \in \ds{F}(2n)$ by the definition below \cref{eq:car}.
Thus, we can perform the FLO twirl on $\mce$ by first performing a Pauli twirl, mapping $\mce \rightarrow \mce^{\ds{P}(n)}$, and then twirling over the full group of FLO unitary transformation, mapping $\mce^{\ds{P}(n)} \rightarrow (\mce^{\ds{P}(n)})^{\ds{F}(2n)} = \mce^{\ds{F}(2n)}$.

By applying the Jordan-Wigner transform to \cref{eq:ptwirl}, together with the first statement of the lemma, we have
\begin{align}
    \mce^{\ds{F}(2n)}(X) &=  \sum_{\bsalph \subseteq [2n]}p_{\mbf{x}(\bsalph)} \int_{\ds{F}(2n)} dU \;\mcu\left(\gamma_{\bsalph}\right) X \mcu(\gamma^{\dagger}_{\bsalph}) \\
    \mce^{\ds{F}(2n)}(X) &= \sum_{\substack{\bsalph \subseteq [2n] \\ \bsbet, \bset | |\bsalph|, 2n}}p_{\mbf{x}(\bsalph)} \left[\int_{O(2n)} d\mbf{R} \det\left(\mbf{R}_{\bsalph \bsbet}\right) \det\left(\mbf{R}_{\bsalph \bset}\right)\right] \left(\mm{\bsbet} X \mm{\bset}^{\dagger}\right)
\end{align}
Since $\mce^{\ds{F}(2n)}$ is invariant under Pauli twirling (again by the argument that every Pauli is a matchgate), we must have that only the terms for which $\bsbet = \bset$ are nonvanishing. 
This gives
\begin{align}
\mce^{\ds{F}(2n)}(X) &= \sum_{\substack{\bsalph \subseteq [2n] \\ \bsbet | |\bsalph|, 2n}}p_{\mbf{x}(\bsalph)} \left[\int_{O(2n)} d\mbf{R} \ |\det\left(\mbf{R}_{\bsalph \bsbet}\right)|^2\right] \left(\mm{\bsbet} \odot \mm{\bsbet}^{\dagger}\right)
\end{align}
To evaluate this integral, note that it is invariant upon making the change-of-variables $\mbf{R} \rightarrow \mbf{S}^{T}\mbf{R}$, where $\mbf{S}$ is a permutation matrix which maps $\bsalph$ to some subset $\bset$ with $|\bset| = |\bsalph|$, as
\begin{align}
    \int_{O(2n)} d\mbf{R} \ |\det\left(\mbf{R}_{\bsalph \bsbet}\right)|^2 &= \int_{O(2n)} d\mbf{R} \ |\det\left[\left(\mbf{S}^{\mathrm{T}} \mbf{R}\right)_{\bsalph \bsbet}\right]|^2\\
    \int_{O(2n)} d\mbf{R} \ |\det\left(\mbf{R}_{\bsalph \bsbet}\right)|^2 &= \int_{O(2n)} d\mbf{R} \ |\det\left(\mbf{R}_{\bset \bsbet}\right)|^2
\end{align}
In this step, we used the Cauchy-Binet formula and the fact that $\det\left[\left(\mbf{S}^{\mathrm{T}}\right)_{\bsalph \boldsymbol{\xi}}\right] = \pm \delta_{\boldsymbol{\xi} \bset}$.
Summing over all such $\bset$ gives
\begin{align}
    \int_{O(2n)} d\mbf{R} \ |\det\left(\mbf{R}_{\bsalph \bsbet}\right)|^2 
    &= \binom{2n}{|\bsalph|}^{-1} \sum_{\bset | |\bsalph|} \int_{O(2n)} d\mbf{R} \ |\det\left(\mbf{R}_{\bset \bsbet}\right)|^2 \\ 
    &= \binom{2n}{|\bsalph|}^{-1} \int_{O(2n)} d\mbf{R} \ \det\left[\left(\mbf{R}^{\mathrm{T}} \mbf{R}\right)_{\bsbet \bsbet}\right] \\
    \int_{O(2n)} d\mbf{R} \ |\det\left(\mbf{R}_{\bsalph \bsbet}\right)|^2 &= \binom{2n}{|\bsalph|}^{-1}
\end{align}
From the first to the second equality, we used the fact that $\det\left(\mbf{R}_{\bset \bsbet}\right) = \det\left(\mbf{R}^{\mathrm{T}}_{\bsbet \bset}\right)$ in one of the factors and used the Cauchy-Binet formula again.
Finally, we evaluated the determinant to 1 and used the normalization of the Haar integral.
This gives 
\begin{align}
    \mce^{\ds{F}(2n)} &= \sum_{\substack{\bsalph \subseteq [2n] \\ \bsbet| |\bsalph|}} \binom{2n}{|\bsalph|}^{-1} p_{\mbf{x}({\bsalph})} \left(\mm{\bsbet} \odot \mm{\bsbet}^{\dagger}\right) \\
    \mce^{\ds{F}(2n)} &= \sum_{k = 0}^{2n} \binom{2n}{k}^{-1} \left(\sum_{\bsalph | k} p_{\mbf{x}(\bsalph)} \right)
    \left( \sum_{\bsbet | k} \mm{\bsbet} \odot \mm{\bsbet}^{\dagger} \right)
\end{align}
Identifying $q_k = \sum_{\bsalph | k} p_{\mbf{x}(\bsalph)}$ gives the second statement of the lemma.
\end{proof}

\subsection{Proof of \texorpdfstring{\cref{lem:bptolambdatransform}}{Lemma 11}}

\begin{replemma}{lem:bptolambdatransform}[Restatement]
For an ensemble of noisy and FLO-twirled FLO circuits $\mc{C}^{\mathds{F}(2n)}_r$,
\begin{align}
    \mbf{P}^{(j)}_{0} = 2^{-n} \mbf{d}^{(n)} \cdot \mbf{M}^{(n)} \cdot \boldsymbol{\Lambda}_{\mathrm{even}}^{(j)}
    \label{eq:ztypeinvrep}
\end{align}
for $r = ``z"$. 
Similarly,
\begin{align}
    \frac{1}{2}(\mbf{P}^{(j)}_{+} + \mbf{P}^{(j)}_{-}) &= 2^{-n} \mbf{d}^{(n - 1)} \cdot \mbf{M}^{(n - 1)} \cdot \bs{\Lambda}_{\mathrm{even} \setminus \{2n\}}^{(j)}
    \label{eq:plustypeinvrep}\\
    \frac{1}{2}(\mbf{P}^{(j)}_+ - \mbf{P}^{(j)}_-) &= 2^{-n} \mbf{d}^{(n - 1)} \cdot \mbf{M}^{(n - 1)} \cdot \bs{\Lambda}_{\mathrm{odd}}^{(j)}
    \label{eq:minustypeinvrep}
\end{align}
for $r = ``x"$.
\end{replemma}

\begin{proof}
We begin with the statement for $\mc{C}_z^{\mathds{F}(2n)}$.
We have
\begin{align}
    \widetilde{\mc{U}}^{(j)}(\ketbra{\mbf{0}}{\mbf{0}}) &= \prod_{m = 1}^{L(j)} \left(\mc{E}_m^{(j)}\right)^{\ds{F}(2n)} \mc{U}^{(j)}_m (\ketbra{\mbf{0}}{\mbf{0}}) \\
    \widetilde{\mc{U}}^{(j)}(\ketbra{\mbf{0}}{\mbf{0}}) &= 2^{-n} \sum_{\mbf{x} \in \mathds{Z}_2^{\times n}}\prod_{m = 1}^{L(j)} \left(\mc{E}_m^{(j)}\right)^{\ds{M}(n)} \mc{U}^{(j)}_m (Z^{\mbf{x}}) \label{eq:zicrcuitexpand}
\end{align}
where $Z^{\mbf{x}} \coloneqq \prod_{j = 1}^n Z_j^{x_j}$. 
Under the Jordan-Wigner transform, we have
\begin{align}
    Z^{\mbf{x}} = \prod_{j = 1}^n\left(-i \gamma_{2j - 1} \gamma_{2j}\right)^{x_j}
\end{align}
Applying the argument leading to \cref{eq:circeigintro} gives
\begin{align}
    \widetilde{\mc{U}}^{(j)}(\ketbra{\mbf{0}}{\mbf{0}}) &= 2^{-n} \sum_{\mbf{x} \in \mathds{Z}_2^{\times n}}\prod_{m = 1}^{L(j)} \xi_{m, 2|\mbf{x}|}^{(j)} \mc{U}^{(j)}_m (Z^{\mbf{x}}) \\
    \widetilde{\mc{U}}^{(j)}(\rho_0) &= 2^{-n} \sum_{\mbf{x} \in \mathds{Z}_2^{\times n}} \Lambda_{2|\mbf{x}|}^{(j)} \mc{U}^{(j)} (Z^{\mbf{x}}) \mathrm{.}
\end{align}
Circuits in $\mc{C}_z^{\mathds{F}(2n)}$ are chosen to interact as the identity on net.
This gives
\begin{align}
    \widetilde{\mc{U}}^{(j)}(\ketbra{\mbf{0}}{\mbf{0}}) &= 2^{-n} \sum_{\mbf{x} \in \mathds{Z}_2^{\times n}}\Lambda_{2|\mbf{x}|}^{(j)} Z^{\mbf{x}}
\end{align}
and
\begin{align}
    \bra{\mbf{w}} \widetilde{\mc{U}}^{(j)}(\ketbra{\mbf{0}}{\mbf{0}}) \ket{\mbf{w}} &= 2^{-n} \sum_{\mbf{x} \in \mathds{Z}_2^{\times n}}\Lambda_{2|\mbf{x}|}^{(j)} \bra{\mbf{w}} Z^{\mbf{x}} \ket{\mbf{w}} \\
    &= 2^{-n} \sum_{\mbf{x} \in \mathds{Z}_2^{\times n}}\Lambda_{2|\mbf{x}|}^{(j)} (-1)^{\mbf{x} \cdot \mbf{w}} \\
    &= 2^{-n} \sum_{k = 0}^n \Lambda_{2k}^{(j)} \sum_{\mbf{x}|k}(-1)^{\mbf{x} \cdot \mbf{w}} \\
    \bra{\mbf{w}} \widetilde{\mc{U}}^{(j)}(\ketbra{\mbf{0}}{\mbf{0}}) \ket{\mbf{w}}&= 2^{-n} \sum_{k = 0}^n \Lambda_{2k}^{(j)} M_{|\mbf{w}| k} \\
\end{align}
Summing over all $\mbf{w}$ with a fixed Hamming weight $\ell$ gives
\begin{align}
    \sum_{\mbf{w}|\ell} \bra{\mbf{w}} \widetilde{\mc{U}}^{(j)}(\ketbra{\mbf{0}}{\mbf{0}}) \ket{\mbf{w}} = 2^{-n} \binom{n}{\ell} \sum_{k = 0}^n M_{\ell k} \Lambda_{2k}^{(j)} 
\end{align}
which is the desired result.
\begin{align}
    \bp^{(j)}_{0} = 2^{-n} \mbf{d}^{(n)} \cdot \mbf{M}^{(n)} \cdot \boldsymbol{\Lambda}_{\mathrm{even}}^{(j)}.
\end{align}
In the $\mc{C}_x^{\mathds{F}(2n)}$-case, instead of \cref{eq:zicrcuitexpand}, we have 
\begin{align}
     \widetilde{\mc{U}}^{(j)}(\ketbra{+}{+}^{\otimes n}) &= \prod_{m = 1}^{L(j)} \left(\mc{E}_m^{(j)}\right)^{\ds{F}(2n)} \mc{U}^{(j)}_m (\ketbra{+}{+}^{\otimes n}) \label{eq:xicrcuitexpand}
\end{align}
As shown in \cite{helsen2022matchgate}, there exists a FLO unitary $\mc{U}_{+}$ effecting
\begin{align}
\mc{U}_{+}(\ketbra{+}{+}^{\otimes n}) = \left(\frac{I + i \gamma_1}{\sqrt{2}}\right) \ketbra{\mbf{0}}{\mbf{0}} \left(\frac{I - i \gamma_1}{\sqrt{2}}\right)
\end{align}
This gives
\begin{align}
    \widetilde{\mc{U}}^{(j)}(\ketbra{+}{+}^{\otimes n}) &= \prod_{m = 1}^{L(j)} \left(\mc{E}_m^{(j)}\right)^{\ds{F}(2n)} \mc{U}^{(j)}_m \left\{\mc{U}^{\dagger}_{+}\left[\left(\frac{I + i \gamma_1}{\sqrt{2}}\right) \ketbra{\mbf{0}}{\mbf{0}} \left(\frac{I - i \gamma_1}{\sqrt{2}}\right)\right]\right\} \\
    &= 2^{-n} \sum_{\mbf{x} \in \mathds{Z}_2^{\times n}}\prod_{m = 1}^{L(j)} \left(\mc{E}_m^{(j)}\right)^{\ds{F}(2n)} \mc{U}^{(j)}_m \left\{\mc{U}^{\dagger}_{+}\left[\left(\frac{I + i \gamma_1}{\sqrt{2}}\right) Z^{\mbf{x}} \left(\frac{I - i \gamma_1}{\sqrt{2}}\right)\right]\right\} \\
    &= 2^{-n} \sum_{\mbf{x} \in \mathds{Z}_2^{\times n}}\prod_{m = 1}^{L(j)} \left(\mc{E}_m^{(j)}\right)^{\ds{F}(2n)} \mc{U}^{(j)}_m \left\{\mc{U}^{\dagger}_{+}\left[(i \gamma_1)^{x_1} Z^{\mbf{x}} \right]\right\} \\
    &= 2^{-n} \sum_{\mbf{x} \in \mathds{Z}_2^{\times n}}\prod_{m = 1}^{L(j)} \xi_{m, 2|\mbf{x}| - x_1}^{(j)} \mc{U}^{(j)}_m \left\{\mc{U}^{\dagger}_{+}\left[(i \gamma_1)^{x_1} Z^{\mbf{x}} \right]\right\} \\ 
    \widetilde{\mc{U}}^{(j)}(\ketbra{+}{+}^{\otimes n}) &= 2^{-n} \sum_{\mbf{x} \in \mathds{Z}_2^{\times n}} \Lambda_{2|\mbf{x}| - x_1}^{(j)} \prod_{m = 1}^{L(j)} \mc{U}^{(j)}_m \left\{\mc{U}^{\dagger}_{+}\left[(i \gamma_1)^{x_1} Z^{\mbf{x}} \right]\right\} \mathrm{.}
\end{align}
Since these circuits act as $\mc{U}_{+}$ on-net, this gives
\begin{align}
     \widetilde{\mc{U}}^{(j)}(\ketbra{+}{+}^{\otimes n}) &= 2^{-n} \sum_{\mbf{x} \in \mathds{Z}_2^{\times n}} \Lambda_{2|\mbf{x}| - x_1}^{(j)} (i \gamma_1)^{x_1} Z^{\mbf{x}}
\end{align}
Measuring in the $Y$-basis on the first qubit and the $Z$-basis on all other qubits, we have
\begin{align}
    \bra{+_y \mbf{y}}\widetilde{\mc{U}}^{(j)}(\ketbra{+}{+}^{\otimes n}) \ket{+_y \mbf{y}} &= 2^{-n} \sum_{\mbf{x} \in \mathds{Z}_2^{\times (n - 1)}} \left( \Lambda^{(j)}_{2|\mbf{x}|} + \Lambda^{(j)}_{2|\mbf{x}| + 1} \right) (-1)^{\mbf{x} \cdot \mbf{y}} \\
    &= 2^{-n} \sum_{k = 0}^{n - 1} \left(  \Lambda^{(j)}_{2k} +  \Lambda^{(j)}_{2k + 1} \right) \sum_{\mbf{x}|k} 
 (-1)^{\mbf{x} \cdot \mbf{y}} \\
    \bra{+_y \mbf{y}}\widetilde{\mc{U}}^{(j)}(\ketbra{+}{+}^{\otimes n}) \ket{+_y \mbf{y}} &= 2^{-n} \sum_{k = 0}^{n - 1} \left(  \Lambda^{(j)}_{2k} +  \Lambda^{(j)}_{2k + 1} \right) M^{(n - 1)}_{|\mbf{y}| k} \mathrm{.}
\end{align}
Summing again over all probabilities of a fixed Hamming weight gives
\begin{align}
     \sum_{\mbf{y}|\ell, n - 1} \bra{+_y \mbf{y}}\widetilde{\mc{U}}^{(j)}(\ketbra{+}{+}^{\otimes n}) \ket{+_y \mbf{y}} &= 2^{-n} \sum_{k = 0}^{n - 1} \left(  \Lambda^{(j)}_{2k} +  \Lambda^{(j)}_{2k + 1} \right)\sum_{\mbf{y} | \ell, n - 1} M^{(n - 1)}_{|\mbf{y}| k} \\
     \sum_{\mbf{y}|\ell, n - 1} \bra{+_y \mbf{y}}\widetilde{\mc{U}}^{(j)}(\ketbra{+}{+}^{\otimes n}) \ket{+_y \mbf{y}} &= 2^{-n} \binom{n - 1}{\ell} \sum_{k = 0}^{n - 1} \left(  \Lambda^{(j)}_{2k} +  \Lambda^{(j)}_{2k + 1} \right) M^{(n - 1)}_{\ell k}
\end{align}
The corresponding probability for the ``$-$'' outcome on the first bit is given by
\begin{align}
    \sum_{\mbf{y}|\ell, n - 1} \bra{-_y \mbf{y}}\widetilde{\mc{U}}^{(j)}(\ketbra{+}{+}^{\otimes n}) \ket{-_y \mbf{y}} &= 2^{-n} \binom{n - 1}{\ell} \sum_{k = 0}^{n - 1} \left(  \Lambda^{(j)}_{2k} -  \Lambda^{(j)}_{2k + 1} \right) M^{(n - 1)}_{\ell k}
\end{align}
and this gives the desired relations
\begin{align}
    \frac{1}{2}(\bp^{(j)}_{+} + \bp^{(j)}_{-}) &= 2^{-n} \mbf{d}^{(n - 1)} \cdot \mbf{M}^{(n - 1)} \cdot \bs{\Lambda}_{\mathrm{even} \setminus \{2n\}}^{(j)}\\
    \frac{1}{2}(\bp^{(j)}_{+} - \bp^{(j)}_{-}) &= 2^{-n} \mbf{d}^{(n - 1)} \cdot \mbf{M}^{(n - 1)} \cdot \bs{\Lambda}_{\mathrm{odd}}^{(j)}
\end{align}
By estimating these probabilities, we can collectively estimate all of the noise eigenvalues.
\end{proof}

\end{document}